\documentclass[12pt]{article}
\usepackage{amsmath}
\usepackage{graphicx,psfrag,epsf}
\usepackage{enumerate}
\usepackage{psfrag,epsf,amsmath}
\usepackage{url} 
\usepackage{amsmath,amssymb,amsfonts,amsthm,mathtools,mathrsfs}
\usepackage{psfrag,epsf,booktabs,float}

\usepackage{float}        
\usepackage{booktabs}     
\usepackage{array}        
\usepackage{makecell}     
\usepackage{colortbl}     
\usepackage{xcolor}       
\usepackage{amsmath}      
\usepackage{amssymb}      
\usepackage{graphicx}     
\usepackage[ruled,vlined]{algorithm2e}

\usepackage{amsmath, amssymb}
\usepackage{booktabs}
\usepackage{multirow}
\usepackage{longtable}
\usepackage{caption,xcolor}

\usepackage{booktabs}  
\usepackage{colortbl}  

\usepackage{setspace}

\usepackage{bm,bbm}
\usepackage{color}
\usepackage{subfigure}
\usepackage[symbol]{footmisc}
\usepackage{scalefnt}
\usepackage{authblk} 
\usepackage{multirow,centernot}
\usepackage[colorlinks=true, citecolor=blue, urlcolor=blue]{hyperref}
\usepackage{natbib}
\usepackage{dsfont}
\usepackage[title]{appendix}
\input cyracc.def

\usepackage[left=1in,top=1.1in,right=0.5in,bottom=1in]{geometry}

\theoremstyle{definition}

\newtheorem{theorem}{Theorem}[section]

\newtheorem{corollary}[theorem]{Corollary}

\newtheorem{definition}{Definition}[section]

\newtheorem{proposition}[theorem]{Proposition}
\newtheorem{remark}[theorem]{Remark}

\makeatletter
\def\@seccntformat#1{\@ifundefined{#1@cntformat}%
	{\csname the#1\endcsname\quad}
	{\csname #1@cntformat\endcsname}
}
\makeatother

\newif\ifShowComments
\ShowCommentstrue
\def\strutdepth{\dp\strutbox}
\def\druk#1{\strut\vadjust{\kern-\strutdepth
        {\vtop to \strutdepth{%
                \baselineskip\strutdepth\vss
                        \llap{\hbox{#1}\quad}\null}}}}

\title{\bf
Bias in estimating Theil, Atkinson, and dispersion indices for gamma mixture populations
}

\author{
\text{Jackson Assis}$^{1}$,
\text{Roberto Vila}$^{1}$
\, and
\text{Helton Saulo}$^{1,2}$\thanks{Corresponding author: Helton Saulo, email: {heltonsaulo@gmail.com}}
\\
{\small $^{1}$ Department of Statistics, University of Brasilia, Brasilia, Brazil}\\
{\small $^{2}$ Department of Economics, Federal University of Pelotas, Pelotas, Brazil}\\
}
\begin{document}
	\maketitle 	
	\begin{abstract}
This paper examines the finite-sample bias of estimators for the Theil and Atkinson indices, as well as for the variance-to-mean ratio (VMR), under the assumption that the population follows a finite mixture of gamma distributions with a common rate parameter. Using Mosimann’s proportion–sum independence theorem and the structural relationship between the gamma and Dirichlet distributions, these estimators were rewritten as functions of Dirichlet vectors, which enabled the derivation of closed-form analytical expressions for their expected values.
A Monte Carlo simulation study evaluates the performance of both the traditional and bias-corrected estimators across a range of mixture scenarios and sample sizes, revealing systematic bias induced by population heterogeneity and demonstrating the effectiveness of the proposed corrections, particularly in small and moderate samples. An empirical application to global per capita GDP data further illustrates the practical relevance of the methodology and confirms the suitability of gamma mixtures for representing structural economic heterogeneity.
\end{abstract}
	\smallskip
	\noindent
	{\small {\bfseries Keywords.} {Inequality; Gamma mixtures; Dirichlet-Gamma structure; Mosimann’s theorem; Theil and Atkinson; Monte Carlo; GDP data.}}
	\\
	{\small{\bfseries Mathematics Subject Classification (2010).} MSC 60E05 $\cdot$ MSC 62Exx $\cdot$ MSC 62Fxx.}
%

Keywords:

\section{Introduction}

The measurement of inequality plays a central role in the applied social sciences, and indices such as those of Gini \citep{Gini1912}, Theil \citep{Theil1967}, and Atkinson \citep{Atkinson1970} have become standard references for quantifying the concentration of income and wealth. The construction of these measures is not neutral, since each index is based on conceptual choices about what should be regarded as inequality and how differences among individuals should be weighted. For this reason, the classical literature focused on discussing the ethical and axiomatic foundations underlying their formulation \citep{sen1973, Atkinson1970}. A systematic presentation of these bases can be found in \cite{Chakravarty2008}. In parallel, the debate gradually began to incorporate inferential aspects, with increasing attention to the behavior of estimators in finite samples and under different forms of population heterogeneity.

Among the indices that explicitly incorporate normative criteria for evaluating inequality, the contributions of Theil and Atkinson stand out, as both embed clear normative principles concerning inequality aversion and sensitivity to different regions of the distribution \citep{bourguignon1979, Cowell2000}. In addition, they satisfy fundamental properties such as invariance to proportional changes and the transfer principle (Pigou–Dalton), which makes them useful for group comparisons and for social welfare analyses \citep{Shorrocks1980}. Their conceptual structures, however, are not identical: while the Theil indices allow rigorous additive decomposition between and within groups, the Atkinson index, although it also permits component analyses, is based on the equally distributed equivalent income and does not admit the same exact form of decomposition.

Despite their desirable properties, these indices involve ratios and nonlinear transformations in their formulation, which makes rigorous statistical analysis in finite samples difficult. Most of the literature assumes homogeneous populations modeled by specific continuous distributions (such as the lognormal, Pareto, or gamma), meaning that the impact of structural heterogeneity or mixture models on estimator performance remains relatively underexplored.

When the hypothesis of population homogeneity is adopted, among the continuous positive distributions used in inequality analysis, the gamma distribution stands out as a particularly convenient choice. Its functional form belongs to the exponential family and admits closed-form moments, which facilitates the analytical treatment of measures based on ratios and logarithmic transformations, such as the Theil and Atkinson indices. Moreover, its shape and rate parameters provide flexibility to represent different degrees of dispersion and skewness, helping explain its extensive empirical use \citep{salem1974, McDonald1979, bourguignon1979}.

However, the simple gamma model presumes population homogeneity, a hypothesis that is often untenable when the phenomenon under study results from the composition of subgroups with distinct characteristics. In such cases, the observed distribution may display multimodality or heavier tails than those captured by a single gamma distribution. As pointed out by \cite{chotikapanich2008}, mixture models of gamma distributions offer a more flexible alternative, preserving parametric interpretability while accommodating structural heterogeneity.

This aspect is particularly relevant in the inferential domain. Even under the simple gamma distribution, recent studies have shown that commonly used estimators of inequality indices—especially those of Gini, Theil, and Atkinson-may exhibit non-negligible bias in finite samples; see \cite{Deltas2003}, \cite{VilaSaulo2025TheilAtkinsonGamma}, \cite{VilaSaulo2025GiniM}, and \cite{Shih2025}. If bias already emerges in homogeneous populations, its magnitude tends to increase in the presence of structural heterogeneity, given the nonlinear nature of these measures. Thus, studying estimator behavior under mixture models becomes a methodological necessity rather than a mere technical extension.

Recently, \cite{VilaSaulo2025GiniMixture} extended the analysis of inequality estimation to heterogeneous populations represented by finite mixtures of Gamma distributions. Building on the analytical approach of \cite{BaydilEtAl2025}, the authors obtained a closed-form formula for the bias of the Gini estimator in this more general setting and showed that heterogeneity among mixture components amplifies sample bias. In the same spirit, \cite{VilaSaulo2025TheilAtkinsonGamma} applied the same analytical approach to study the bias of the Theil and Atkinson index estimators under the assumption of a classical Gamma distribution. The authors obtained closed-form expressions for the respective biases and proposed explicit bias-corrected estimators.

Motivated by this recent progress, the present work contributes to the ongoing analytical characterization of bias. Its aim is to investigate the behavior of the estimators of the Theil and Atkinson indices and of the dispersion index (VMR) in populations whose distribution is represented by a finite mixture of Gamma distributions with a common rate parameter. On the theoretical front, probabilistic tools based on Mosimann’s proportion–sum independence theorem and on properties of the Dirichlet distribution are employed to derive closed-form expressions for the expected value of the estimators and, consequently, to analytically assess the bias introduced by the structural heterogeneity of the mixture. Monte Carlo simulation experiments are conducted to empirically evaluate the performance of the sample estimators and their corrected versions, which follow directly from the analytical bias expressions. An empirical application is developed using real socioeconomic data (GDP per capita), in which a gamma mixture distribution is fitted to the dataset. This step allows assessment of the adequacy of the proposed models and of the practical relevance of the bias corrections, complementing the theoretical and simulation analyses. The results highlight the applicability of the corrections derived and their usefulness in real socioeconomic contexts.

The rest of this paper unfolds as follows. Section \ref{sec:02} presents the necessary preliminaries, including notation, definitions, and key probabilistic results used throughout the paper. Section \ref{sec:theil-mixture-gama} derives the population versions of the Theil, Atkinson, and VMR indices under finite gamma mixture models. In Section \ref{Deriving estimator biases}, we establish closed-form analytical expressions for the finite-sample biases of the corresponding estimators. Section \ref{cap:simulacao} reports a Monte Carlo simulation study assessing the empirical performance of the estimators and their bias-corrected versions. Section \ref{cap:dados-reais} then presents an empirical application to global per capita GDP data, illustrating the practical relevance of the proposed corrections. Finally, Section \ref{concluding_remarks} provides concluding remarks and discusses potential directions for future research.

\section{Preliminary results and some definitions}\label{sec:02}

\begin{definition}[Mixture of gammas] \label{definition:mixture-gamas}
	A random variable $X$ follows a gamma mixture distribution with parameter vector $\bm{\theta} = (\pi_1, \dots, \pi_m, \alpha_1, \dots, \alpha_m, \lambda)^\top$, denoted by $X \sim \text{GM}(\bm{\theta})$, if its probability density function is given by
	\begin{equation}
		\label{eq:mixture-gama}
		f_X(x;\bm{\theta}) = \sum_{j=1}^{m} \pi_j f_{Z_j}(x;\alpha_j,\lambda), \quad x > 0,
	\end{equation}
	where $m \in \mathbb{N}$ is the number of components, $\pi_j$ are the mixing proportions, with $\sum_{j=1}^{m} \pi_j = 1$ and $\pi_j > 0$. Moreover, $f_{Z_j}(x;\alpha_j,\lambda)$ is the density of the gamma-distributed random variable $Z_j$ with shape parameter $\alpha_j>0$ and rate $\lambda>0$, denoted by $Z_j \sim \text{Gamma}(\alpha_j,\lambda)$, that is,
	\begin{equation*}
		f_{Z_j}(x;\alpha_j,\lambda)
		=
		\frac{\lambda^{\alpha_j}}{\Gamma(\alpha_j)}\,
		x^{\alpha_j - 1} e^{-\lambda x}, \quad x > 0,
	\end{equation*}
	where $\Gamma(\cdot)$ denotes the (complete) gamma function.
\end{definition}

Figure 1 of \cite{VilaSaulo2025GiniMixture} illustrates how the gamma mixture distribution acts as a flexible framework that includes numerous important special cases.
\begin{remark}
	Note that for $m=2$, with $\boldsymbol{\theta}' = (\pi_1,\pi_2,\alpha_1,\alpha_2,\lambda)^\top$ and $\boldsymbol{\theta}'' = (\pi_2,\pi_1,\alpha_2,\alpha_1,\lambda)^\top$, we have $f(x;\bm{\theta}') = f(x;\bm{\theta}'')$ for all $x > 0$. In other words, the distribution of gamma mixtures is not identifiable. This lack of identifiability may lead to unstable parameter estimates, convergence issues in estimation algorithms, multimodality in the likelihood function, and misleading interpretations regarding the mixture components' structure.
	To address these challenges, authors such as \cite{mclachlan2000finite} and \cite{fruhwirth2006finite} propose strategies like imposing parameter constraints, adopting Bayesian approaches, and applying relabeling methods to realign the samples generated by MCMC algorithms and ensure more consistent inference.
\end{remark}

It is clear that $X$ in Definition~\ref{eq:mixture-gama} has the following representation
\begin{align}\label{representation-X}
	X
	=
	\sum_{j=1}^{m}\mathds{1}_{\{Y=j\}}Z_j,
\end{align}
where $Y\in\{1,\ldots,m\}$ is a discrete random variable  with $\mathbb{P}(Y=j)=\pi_j$, independent of $Z_j\sim{\rm Gamma}(\alpha_j,\lambda)$, for each $j=1,\ldots,m$, and $\mathds{1}_A$ is the indicator function of the event $A$.

Since the sequence of events $\{\{Y=j\}:j=1,\ldots, m\}$ is a partition of the sample space, applying the Law of total probability and the independence of $Y$ and $Z_j$, the  cumulative distribution function (CDF) of $X$ can be written as
\begin{align*}
	F_X(x;\boldsymbol{\theta})
	=
	\sum_{j=1}^{m}
	\pi_j
	F_{Z_j}(x;\alpha_j,\lambda).
\end{align*}

As a consequence, the following result is obtained:
\begin{proposition}
	\label{proposition:ev-gx}
	If $X \sim \text{GM}(\bm{\theta})$, then, for any Borel-measurable function $h$ we have
	\[
	\mathbb{E}[h(X)]
	=
	\sum_{j=1}^m\pi_j\mathbb{E}[h(Z_j)],
	\]
	where
	$Z_j\sim{\rm Gamma}(\alpha_j,\lambda)$, $j=1,\ldots,m$.
\end{proposition}

\begin{proposition}\label{proposition:ev-gx-1}
	If $X_1,\ldots,X_n$ is a random sample (i.i.d.) of size $n$ from $X \sim \text{GM}(\bm{\theta})$, then, for any Borel-measurable function $h$, we have
	\begin{align*}
		\mathbb{E}\left[h\left(\sum_{i=1}^{n}X_i\right)\right]
		=
		\sum_{1\leqslant j_1,\ldots,j_n\leqslant m} \pi_{j_1}\cdots \pi_{j_n}\:
		\mathbb{E}\left[h\left(Z^*_{j_1,\ldots,j_n}\right)\right],
	\end{align*}
	where $Z^*_{j_1,\ldots,j_n}\sim{\rm Gamma}(\sum_{i=1}^n\alpha_{j_i},\lambda)$.
\end{proposition}

\begin{proof}
	The proof is immediate since,  by the i.i.d. nature of the variables $X_1,\ldots,X_n$, for any Borel-measurable function $h$,  we have
	\begin{align*}
		\mathbb{E}\left[h\left(\sum_{i=1}^{n}X_i\right)\right]
		&=
		\int_{(0,\infty)^n}
		h\left(\sum_{i=1}^{n}x_i\right) f_{X_1,\ldots,X_n}(\boldsymbol{x}) {\rm d}\boldsymbol{x}
		\nonumber
		\\[0,2cm]
		& =
		\sum_{1\leqslant j_1,\,\dots,\, j_n\leqslant m}
		\pi_{j_1}\cdots \pi_{j_n}
		\int_{(0,\infty)^n} h\left(\sum_{i=1}^{n}x_i\right) f_{Z_{j_1}}(x_1)\cdots f_{Z_{j_n}}(x_n) {\rm d}\boldsymbol{x}
		\nonumber
		\\[0,2cm]
		& =
		\sum_{1\leqslant j_1,\,\dots,\, j_n\leqslant m}
		\pi_{j_1}\cdots \pi_{j_n}
		\mathbb{E}\left[h\left(\sum_{i=1}^{n}Z_{j_i}\right)\right],
	\end{align*}
	where $Z_{j_1},\ldots,Z_{j_n}$ are independent gamma random variables with parameter vectors $(\alpha_{j_1},\lambda)^\top,\ldots,(\alpha_{j_n},\lambda)^\top$, respectively. By defining $Z^*_{j_1,\ldots,j_n}\equiv\sum_{i=1}^{n}Z_{j_i}$, the proof follows.
\end{proof}

\begin{proposition}
	\label{proposition:mean-variance}
	If $X \sim \text{GM}(\bm{\theta})$, then its $p$-th moment is given by
	\[
	\mathbb{E}[X^p]
	=
	{1\over\lambda^p}
	\sum_{j=1}^m\pi_j \,{\Gamma(p+\alpha_j)\over\Gamma(\alpha_j)},
	\quad
	p>-\alpha_{(1)},
	\]
	where $\alpha_{(1)}=\min\{\alpha_1,\ldots,\alpha_m\}$ and $\Gamma(\cdot)$ is the gamma function.

	In particular, we have
	\begin{align*}
		&\mu = \mathbb{E}[X]
		=
		\frac{1}{\lambda}
		\sum_{j=1}^{m} \pi_j \alpha_j,
		\\[0,2cm]
		&\sigma^2
		=
		\text{Var}(X)
		=
		\frac{1}{\lambda^2}
		\left[
		\sum_{j=1}^{m}
		\pi_j
		\alpha_j(\alpha_j+1)
		-
		\left(
		\sum_{j=1}^{m} \pi_j \alpha_j\right)^2
		\right].
	\end{align*}
\end{proposition}
\begin{proof}
	The proof follows directly from Proposition~\ref{proposition:ev-gx} with $h(x)=x^p$ plus the fact that $Z_j\sim{\rm Gamma}(\alpha_j,\lambda)$.
\end{proof}

{
	In Statistics, Mosimann's proportion-sum independence theorem \citep{Mosimann1962} is a key result in the study of proportions, especially in relation to the Dirichlet distribution.

	\begin{theorem}[\cite{Mosimann1962}] \label{Mosimann}
		Let $\{Y_j:j=1,\ldots,n\}$
		be $n$ positive and mutually independent random variables. Define the random variable $D_j= Y_j/\sum_{i=1}^n Y_i$ for $j= 1,\ldots,n$. Then,
		each $D_j$ is distributed independently of $\sum_{i=1}^n Y_i$ iff $(D_1,\ldots,D_n)^\top\sim \text{Dirichlet}(\boldsymbol{a})$ on $\{(d_1,\ldots,d_n)^\top:d_j>0,1\leqslant j\leqslant n, \sum_{j=1}^n d_j=1\}$.
	\end{theorem}

}

{
	The next technical result will be crucial for the derivation of estimator biases presented in Section \ref{Deriving estimator biases}.
	\begin{proposition} \label{proposition:moments-dirichlet-0}
		For \((D_1,\ldots,D_n)^\top \sim \text{Dirichlet}(\alpha_1, \ldots, \alpha_{n}) \), we have
		\begin{multline}
			\mathbb{E}\left[
			\left(\prod_{j=1}^n D_j^{c_j}\right)^r \log\left(\prod_{j=1}^n D_j^{c_j}\right)
			\right]
			=
			\frac{\Gamma(\sum_{l=1}^n \alpha_l)}{\Gamma\bigl(\sum_{l=1}^n (\alpha_l+ rc_l)\bigr)}
						\\[0,2cm]
						\times
			\prod_{j=1}^{n} \frac{\Gamma(\alpha_{j}+rc_j)}{\Gamma(\alpha_{j})}
			\left[
			c_j\psi(\alpha_{j}+rc_j)
			-
			\sum_{l=1}^n c_l
			\psi\left(\sum_{l=1}^n (\alpha_l+ rc_l)\right)
			\right],          \label{identity:expected-value-prod-pi-0}
		\end{multline}
		where $c_1,\ldots,c_n\geqslant0$, $r\geqslant0$ and $\psi(\cdot)$ denotes the digamma function.
	\end{proposition}
	\begin{proof}
		By using the identity
		$$
		x^r\log(x)
		=
		\lim_{p\to 0} {\partial x^{p+r}\over \partial p},
		\quad x>0, \ r\geqslant 0,
		$$
		we can write
		\begin{align}\label{main-id}
			\mathbb{E}\left[
			\left(\prod_{j=1}^n D_j^{c_j}\right)^r \log\left(\prod_{j=1}^n D_j^{c_j}\right)
			\right]
			=
			\lim_{p\to 0}
			{\partial \over\partial p}
			\mathbb{E}\left[
			\prod_{j=1}^n D_j^{(p+r)c_j}
			\right],
		\end{align}
		where the application of the dominated convergence theorem permitted the exchange of the order of expectation, limit, and differentiation.
		Applying the well-established formula for the mixed moment of $(D_1,\ldots,D_n)^\top$ \citep[see Item (2.5) of][]{Kai2011}:
		\begin{align}
			\mathbb{E} \left[\prod_{j=1}^n D_j^{d_j} \right]
			=
			\frac{\Gamma(\sum_{l=1}^n \alpha_l)}{\Gamma\bigl(\sum_{l=1}^n (\alpha_l+ d_{l})\bigr)} \prod_{j=1}^{n} \frac{\Gamma(\alpha_{j}+d_{j})}{\Gamma(\alpha_{j})},
			\label{identity:expected-value-prod-pi}
		\end{align}
		the expectation on the right-hand side of \eqref{main-id} can be written as
		\begin{align*}
			\mathbb{E}\left[
			\prod_{j=1}^n D_j^{(p+r)c_j}
			\right]
			=
			\frac{\Gamma(\sum_{l=1}^n \alpha_l)}{\Gamma\bigl(\sum_{l=1}^n (\alpha_l+ (p+r)c_l)\bigr)} \prod_{j=1}^{n} \frac{\Gamma(\alpha_{j}+(p+r)c_j)}{\Gamma(\alpha_{j})}.
		\end{align*}
		Consequently, the corresponding partial derivative with respect to $p$ yields:
		\begin{multline}\label{der-crt-p}
			{\partial\over\partial p}
			\mathbb{E}\left[
			\prod_{j=1}^n D_j^{(p+r)c_j}
			\right]
			=
			\frac{\Gamma(\sum_{l=1}^n \alpha_l)}{\Gamma\bigl(\sum_{l=1}^n (\alpha_l+ (p+r)c_l)\bigr)}
			\\[0,2cm]
			\times
			\prod_{j=1}^{n} \frac{\Gamma(\alpha_{j}+(p+r)c_j)}{\Gamma(\alpha_{j})}
			\left[
			c_j\psi(\alpha_{j}+(p+r)c_j)
			-
			\sum_{l=1}^n c_l
			\psi\left(\sum_{l=1}^n (\alpha_l+ (p+r)c_l)\right)
			\right].
		\end{multline}
		Letting $p\to 0$ in \eqref{der-crt-p} and substituting into \eqref{main-id} completes the proof of Identity \eqref{identity:expected-value-prod-pi-0}.
	\end{proof}
}

\section{Theil, Atkinson and dispersion indices
}
\label{sec:theil-mixture-gama}

\subsection{Theil index}

The Theil {(population)} index is a statistical measure of economic inequality, quantifying the divergence from a perfectly equal distribution of income.

\begin{definition}
	\label{def-1-1-1}
	The Theil $T$ index \citep{Theil1967} of a random variable $X$ with finite mean $\mathbb{E}[X]=\mu$ is defined as
	\begin{align*}
		T_T= {\mathbb{E}\left[{X\over\mu}\, \log\left({X\over\mu}\right)\right]},
	\end{align*}
	and the Theil $L$ index is defined as
	\begin{align*}
		T_{L} = \mathbb{E} \left[ \log\left({\frac{\mu}{X}}\right) \right].
	\end{align*}
\end{definition}

\begin{proposition}
	\label{proposition:theil-tt}
	The Theil $T$ index of $X \sim \text{GM}(\bm{\theta})$ is given by
	\[ T_T
	=
	\frac{1}{\sum_{i=1}^m\pi_j\alpha_j}
	\left[
	\sum_{j=1}^{m}
	\pi_j\alpha_j \psi(\alpha_j)
	+
	1
	\right]
	-
	\log\left(\sum_{i=1}^m\pi_j\alpha_j\right),
	\]
	where
	$\psi(\cdot)$ denotes the digamma function.
\end{proposition}
\begin{proof}
	Note that $T_T$ in Definition \ref{def-1-1-1} can be written as
	\begin{align}\label{def_tt}
		T_T
		=
		\frac{1}{\mu} \,
		\mathbb{E}\bigl[X \log(X) \bigr] - \log (\mu).
	\end{align}
	From Proposition \ref{proposition:ev-gx} we have
	\begin{align}\label{def_tt-10}
		\mathbb{E}[X \log (X)]
		=
		\sum_{j=1}^m\pi_j\mathbb{E}[Z_j\log(Z_j)],
	\end{align}
	where
	$Z_j\sim{\rm Gamma}(\alpha_j,\lambda)$, for each $j=1,\ldots,m$.
	As $\mathbb{E}[Z\log(Z)]=(a/b)[\psi(a)+1/a-\log(b)$  for $Z\sim{\rm Gamma}(a,b)$, and $\mu = \sum_{j=1}^{m} \pi_j \alpha_j/{\lambda}$ (Proposition \ref{proposition:mean-variance}), by combining the above Identity \eqref{def_tt-10} with \eqref{def_tt}, the proof follows.
\end{proof}

\begin{proposition}
	\label{proposition:theil-tl}
	The Theil $L$ index of $X \sim \text{GM}(\bm{\theta})$ is given by
	\[
	T_L
	=
	\log\left(\sum_{i=1}^m\pi_j\alpha_j\right)
	-
	\sum_{j=1}^{m}\pi_j \psi(\alpha_j),
	\]
	where $\psi(\cdot)$ is the digamma function.
\end{proposition}

\begin{proof}
	As $T_L
	=\log(\mu) - \mathbb{E}[\log (X)]$,
	by applying  Proposition \ref{proposition:ev-gx} we have
	\begin{align}\label{def_tt-1}
		\mathbb{E}[\log (X)]
		=
		\sum_{j=1}^m\pi_j\mathbb{E}[\log(Z_j)],
	\end{align}
	where
	$Z_j\sim{\rm Gamma}(\alpha_j,\lambda)$, for each $j=1,\ldots,m$. As $\mathbb{E}[\log(Z)]=\psi(a)-\log(b)$  for $Z\sim{\rm Gamma}(a,b)$, and
	$\mu
	=
	\sum_{j=1}^{m} \pi_j \alpha_j/{\lambda}$ (Proposition \ref{proposition:mean-variance}), from the above identities the proof readily  follows.
\end{proof}

\begin{remark}
	Note that, unlike \cite{VilaSaulo2025TheilAtkinsonGamma}, it is not possible to  express $T_L$ as a function of $T_T$, since the digamma function is non-linear.
\end{remark}

\subsection{Atkinson index}

The Atkinson {(population)} index \citep{Atkinson1970} of a random variable \(X\) with finite mean $\mathbb{E}[X]=\mu$ is defined as
\begin{align}\label{i-1}
	A(\varepsilon)
	=
	1 - \frac{\mathbb{E}^{{1}/{(1-\varepsilon)}}[X^{1-\varepsilon}]}{\mu},
	\quad
	0\leqslant \varepsilon \neq 1.
\end{align}
{
	In this subsection we derive explicit forms for the Atkinson  index in the special cases $\varepsilon\to 1$ and $\varepsilon\to \infty$ with $X \sim \text{GM}(\bm{\theta})$.}

\begin{proposition}\label{pro-atkinson-1}
	When $\varepsilon\to 1$, the Atkinson index of $X \sim \text{GM}(\bm{\theta})$ is given by
	\[
	A(1) \equiv
	\lim_{\varepsilon \to 1} A(\varepsilon)
	=
	1 - \frac{1}{\sum_{j=1}^m \pi_j \alpha_j} \exp\left\{\sum_{j=1}^m \pi_j \psi(\alpha_j)\right\},
	\]
	where $\psi(\cdot)$ is the digamma function.
\end{proposition}
\begin{proof}
	{Simple calculations show that}
	\begin{align}
		\lim_{\varepsilon \to 1} \mathbb{E}^{1/(1-\varepsilon)}[X^{1-\varepsilon}]
		&=
		\exp\left\{
		\lim_{\varepsilon \to 1} \frac{1}{1-\varepsilon} \log \left(\mathbb{E}[X^{1-\varepsilon}]\right)
		\right\}
		\nonumber \\[0.2cm]
		&=
		\exp\left\{
		- \lim_{\varepsilon \to 1} \frac{\partial}{\partial\varepsilon} \log \left(\mathbb{E}[X^{1-\varepsilon}]\right)
		\right\}
		\nonumber \\[0.2cm]
		&=
		\exp\left\{
		- \lim_{\varepsilon \to 1} \frac{\partial}{\partial\varepsilon} \mathbb{E}[X^{1-\varepsilon}]
		\right\}
		\nonumber \\[0.2cm]
		&=
		\exp\left\{
		\mathbb{E}[\log(X)]
		\right\},
		\label{limit-1}
	\end{align}
	where in the second equality we used L'Hôpital's rule, and in the fourth equality the identity
	\[
	\frac{\partial}{\partial\varepsilon} x^{1-\varepsilon} = -x^{1-\varepsilon} \log(x), \quad x>0.
	\]

	Consequently, combining \eqref{i-1} with \eqref{limit-1}, we obtain
	\[
	A(1) = \lim_{\varepsilon \to 1} A(\varepsilon) = 1 - \frac{\exp\left\{\mathbb{E}[\log(X)]\right\}}{\mu}.
	\]
	Since, by \eqref{def_tt-1}, \(\mathbb{E}[\log (X)] = \sum_{j=1}^m \pi_j \left[\psi(\alpha_j) - \log(\lambda)\right]\) and \(\mu =\sum_{j=1}^m \pi_j \alpha_j/\lambda\) (Proposition \ref{proposition:mean-variance}), the proof follows.
\end{proof}

\begin{proposition}\label{pro-atkinson-infty}
	When $\varepsilon\to \infty$, the Atkinson index of $X \sim \text{GM}(\bm{\theta})$ is given by
	\[
	A(\infty) {\equiv }  \lim_{\varepsilon \to \infty} A(\varepsilon)
	= 1.
	\]
\end{proposition}
\begin{proof}
	Note that the Atkinson index \eqref{i-1} can be written as
	\begin{align*}
		A(\varepsilon)
		=
		1
		-
		\frac{
			\mathbb{E}^{-{1}/{(\varepsilon-1)}}\left[\dfrac{1}{X^{\varepsilon-1}}\right]
		}{
			\mu
		},
		\quad	0\leqslant \varepsilon \neq 1.
	\end{align*}
	Using the norm notation in \(L^p\) spaces (Chapter 4 of \cite{Brezis2010}):
	\[
	\|Y\|_p \equiv\mathbb{E}^{1/p}[Y^p] = \left[\int_0^\infty y^p \, \mathrm{d}F_Y(y)\right]^{1/p},
	\quad
	p>0,
	\]
	the index \(A(\varepsilon)\) can be expressed as
	\begin{align}\label{i-2}
		A(\varepsilon)
		=
		1 - \frac{\displaystyle {\left\| \frac{1}{X} \right\|_{\varepsilon-1}^{-1}}}{\mu},
		\quad	0\leqslant \varepsilon \neq 1.
	\end{align}
	Using the well-known fact: \(\lim_{p \to \infty} \|Y\|_p = \|Y\|_\infty\) (see Exercise 4.6 of Chapter 4 in \cite{Brezis2010}), where
	\[
	\|Y\|_\infty \equiv \inf\{M \geqslant 0 : \mathbb{P}(|Y| \leqslant M)=1\}
	\]
	is the essential supremum, we have
	\begin{align}\label{limit-2}
		\lim_{\varepsilon \to \infty}
		\left\| \frac{1}{X} \right\|_{\varepsilon - 1}
		=
		\left\| \frac{1}{X} \right\|_{\infty}
		&=
		\inf \left\{ M \geqslant 0 : \mathbb{P}\left(\left| \frac{1}{X} \right| \leqslant M\right)=1 \right\}
		\nonumber
		\\[0,2cm]
		&
		=
		\inf \left\{ M \geqslant 0 : F_X\left( \frac{1}{M};\boldsymbol{\theta}\right)=0\right\}
		=
		\infty.
	\end{align}
	Consequently, combining \eqref{i-2} with \eqref{limit-2}, the proof follows.
\end{proof}

\subsection{Dispersion index}

In probability theory and statistics, the index of dispersion, or variance-to-mean ratio (VMR), is a normalized measure of variability used to assess whether observed occurrences are more clustered or more dispersed than expected under a standard statistical model

Proposition \ref{proposition:mean-variance} implies the following result.
\begin{corollary}
	\label{corollary:vmr}
	If $X \sim \text{GM}(\bm{\theta})$, then the ratio of the variance to the mean of the mixture of gammas with fixed $\lambda>0$ is given by
	\[\text{VMR}
	=
	\text{VMR}[X]
	\equiv
	\frac{\sigma^2}{\mu}
	=
	\dfrac{1}{
		\lambda \sum_{j=1}^m \pi_j \alpha_j}
	\left[
	{ \displaystyle
		\sum_{j=1}^m \pi_j \alpha_j(\alpha_j + 1) - \left( \sum_{j=1}^m \pi_j \alpha_j \right)^2
	}
	\right].
	\]
\end{corollary}

\section{Deriving estimator biases}\label{Deriving estimator biases}

\subsection{Bias of the Theil  index estimator}
\label{subsec:expected-value-theil}

The next results provide closed-form expressions for the expected value of the Theil index estimators $\widehat{T}_T$ and $\widehat{T}_L$, given by
\begin{align*}
	\widehat{T}_T
	=
	\dfrac{
		\displaystyle
		\sum_{i=1}^{n}
		X_{i}
		\log\left(\frac{X_i}{\overline{X}}\right)
	}{\displaystyle
		\sum _{i=1}^{n} X_i},
	\quad
	n\in\mathbb{N},
\end{align*}
and
\begin{align*}
	\widehat{T}_{L}
	=
	{\frac{1}{n}}
	\sum_{i=1}^{n}
	\log\left({\frac{\overline{X}}{X_{i}}}\right),
	\quad
	n\in\mathbb{N},
\end{align*}
respectively, where $\overline{X}=\sum_{i=1}^{n}X_i/n$ is the sample mean and $X_1,\ldots,X_n$ are i.i.d. observations from \(X\sim \mathrm{GM}(\bm{\theta})\).

\begin{remark}\label{remark:rewrite-tt}
	Note that $\widehat{T}_T$ can be written as
	\begin{align*}
		\widehat{T}_T
		=
		\sum_{i=1}^n D_i \log(D_i) + \log(n),
		\quad
		D_i = {X_i\over S_n},
		\quad
		S_n = \sum_{i=1}^n X_i.
	\end{align*}
\end{remark}

\begin{remark}
	We can write $\widehat{T}_L$ as
	\begin{align}\label{remark:rewrite-tl}
		\widehat{T}_L = \log(\overline{X}) - \frac{1}{n} \sum_{i=1}^n \log(X_i).
	\end{align}
\end{remark}

\begin{theorem}
	\label{theorem:expected-value-tt}
	If \(X_1,\dots,X_n\) are independent copies of \(X\sim \mathrm{GM}(\bm{\theta})\), according to the Definition \ref{definition:mixture-gamas}, then
	\begin{multline*}
	\mathbb{E}\left[\hat T_T\right]
	=
	\sum_{1\leqslant j_1,\, \dots,\, j_n\leqslant m}
	\pi_{j_1} \cdots \pi_{j_n} \,
	\frac{1}{\sum_{i=1}^n \alpha_{j_i}}
	\\[0,2cm]
	\times
	\left[
	\sum_{i=1}^n \alpha_{j_i} \psi(\alpha_{j_i})
	-
	\left( \sum_{i=1}^n \alpha_{j_i} \right) \psi\left( \sum_{i=1}^n \alpha_{j_i}\right)
	+
	n-1
	\right]
	+ \log(n),
	\end{multline*}
	where
	$\psi(\cdot)$ is the digamma function.
\end{theorem}
\begin{proof}
	Since $X_1,\ldots,X_n$ are i.i.d. with the same distribution of $X \sim \text{GM}(\bm{\theta})$, and $X$ has the stochastic representation \eqref{representation-X}, then it is clear that $X_i$ admits the following representation:
	\begin{align*}
		X_i
		\equiv
		\sum_{j=1}^{m}\mathds{1}_{\{Y_i=j\}}Z_j^i,
		\quad
		i=1,\ldots,n,
	\end{align*}
	where $Y_1,\ldots,Y_n\in \{1,\ldots,m\}$ are i.i.d. discrete random variables  with $\mathbb{P}(Y_i=j)=\pi_j$, independent of $Z_j^i\sim{\rm Gamma}(\alpha_j,\lambda)$, for each $j=1,\ldots,m$.

	By using the notations adopted in Remark \ref{remark:rewrite-tt}, note that
	\begin{align}
		S_n \,\vert \, Y_1=j_1,\ldots,Y_n=j_n
		\stackrel{d}{=}
		\sum_{i=1}^n
		Z_{j_i}^i
		\sim{\rm Gamma}\left(\sum_{i=1}^{n}\alpha_{j_i},\lambda\right),
	\end{align}
	with $\stackrel{d}{=}$ meaning equality in distribution of random variables.
	%

	At this point, by defining \( D_i = X_i / S_n \), it follows from Theorem \ref{Mosimann} that, given \( Y_1 = j_1, \ldots, Y_n = j_n \), \( S_n \) and \((D_1, \ldots, D_n)^\top\) are independent. Furthermore,
	\begin{align}\label{cond-dirichlet}
		(D_1, \ldots, D_n)^\top \mid Y_1 = j_1, \ldots, Y_n = j_n \sim \text{Dirichlet}(\alpha_{j_1}, \ldots, \alpha_{j_n}).
	\end{align}

	Therefore, by using Remark \ref{remark:rewrite-tt}, the expected value of the estimator $\widehat{T}_T$, conditioned on $Y_1=j_1,\ldots,Y_n=j_n$, is
	\begin{align*}
	\mathbb{E}\left[\widehat{T}_T \, \Big\vert \,Y_1=j_1,\ldots,Y_n=j_n\right]
	=
	\sum_{i=1}^n \mathbb{E}[D_i \log(D_i) \mid Y_1=j_1,\ldots,Y_n=j_n] + \log(n).
	\end{align*}
	Since \eqref{cond-dirichlet} is satisfied, we use  Proposition \ref{proposition:moments-dirichlet-0} with $r=1$ and $c_j=\delta_{ji}$, where $\delta_{ji}$ is the Kronecker delta, to obtain
	\begin{align*}
	\mathbb{E}\left[\widehat{T}_T \, \Big\vert \,Y_1=j_1,\ldots,Y_n=j_n\right]
	=
	\frac{1}{\sum_{i=1}^{n}\alpha_{j_i}}
	\sum_{i=1}^n
	\alpha_{j_i}
	\left[ \psi(\alpha_{j_i}+1) - \psi\left(\sum_{i=1}^{n}\alpha_{j_i}+1\right) \right]
	+ \log(n).
	\end{align*}
	Finally, by using the well-known recurrence relation $\psi(x+1)=\psi(x)+1/x$ and, then, by applying the Law of total expectation the expression of the statement is obtained.
\end{proof}
Combining Proposition~\ref{proposition:theil-tt} and Theorem~\ref{theorem:expected-value-tt}, we obtain the following result.

\begin{corollary} 
	The bias of \(\widehat{T}_T\) relative to \(T_T\), denoted by \(\mathrm{Bias}(\widehat{T}_T, T_T)\), can be expressed as
	\begin{multline*}
		\mathrm{Bias}(\widehat{T}_T, T_T)
		=
		\sum_{1 \leqslant j_1, \dots, j_n \leqslant m}
		\pi_{j_1} \cdots \pi_{j_n}
		\,
		\frac{1}{\sum_{i=1}^n \alpha_{j_i}}
		\\[0,2cm]
		\times
		\left[
		\sum_{i=1}^n \alpha_{j_i} \psi(\alpha_{j_i})
		-
		\left( \sum_{i=1}^n \alpha_{j_i} \right) \psi\left( \sum_{i=1}^n \alpha_{j_i} \right)
		+
		n-1
		\right]
				+ \log(n)
		\\[0,2cm]
		- \left[
		\frac{1}{\sum_{j=1}^{m} \pi_j \alpha_j} \left(
		\sum_{j=1}^{m} \pi_j \alpha_j \psi(\alpha_j)
		+ 1 \right)
		- \log\left(\sum_{j=1}^{m} \pi_j \alpha_j \right) \right],
	\end{multline*}
	where \(\psi(\cdot)\) is the digamma function.
\end{corollary}

{
	\begin{remark}
		In the case where \(\alpha_j = \alpha \text{ for all } j=1,\ldots,m\), the corresponding bias is:
		\begin{align*}
			\mathrm{Bias}(\widehat{T}_T, T_T)
			=
			\log(n\alpha)
			-
			\frac{1}{n\alpha}
			-
			\psi\left(n \alpha\right)
			<0,
		\end{align*}
		which supports the recent finding stated in Corollary 3.6 of \cite{VilaSaulo2025TheilAtkinsonGamma}.
	\end{remark}
}

\begin{theorem}
	\label{theorem:expected-value-tl}
	If $X_1, \dots, X_n$ are independent copies of $X \sim \operatorname{GM}(\bm{\theta})$, then
	\[
	\mathbb{E}\left[\widehat{T}_L\right]
	=
	\sum_{1\leqslant j_1,\, \dots,\, j_n\leqslant m}
	\pi_{j_1}\cdots\pi_{j_n}
	\psi\left(\sum_{i=1}^n\alpha_{j_i}\right)
	-
	\log(n)
	-
	\sum_{j=1}^m \pi_j\psi(\alpha_j),
	\]
	where
	$\psi(\cdot)$ denotes the digamma function.
\end{theorem}

\begin{proof}

	By using Identity \eqref{remark:rewrite-tl} and the fact that $X_1, \dots, X_n$  are i.i.d. with $X$, we have
	\begin{align*}
		\mathbb{E}\left[\widehat{T}_L\right]
		=
		\mathbb{E}\left[
		\log\left(\sum_{i=1}^n X_i\right)
		\right]
		-
		\log(n)
		-
		\mathbb{E}[\log(X)].
	\end{align*}
	Applying Propositions \ref{proposition:ev-gx} and \ref{proposition:ev-gx-1}, the identity above can be written as follows:
	\begin{align*}
		\mathbb{E}\left[\widehat{T}_L\right]
		=
		\sum_{1\leqslant j_1,\, \dots,\, j_n\leqslant m}
		\pi_{j_1}\cdots\pi_{j_n}
		\mathbb{E}\left[
		\log\left(Z^*_{j_1,\ldots,j_n}\right)
		\right]
		-
		\log(n)
		-
		\sum_{j=1}^m
		\pi_j
		\mathbb{E}[\log(Z_j)].
	\end{align*}
	As $Z^*_{j_1,\ldots,j_n}\sim{\rm Gamma}(\sum_{i=1}^n \alpha_{j_i},\lambda)$ and $Z_{j}\sim{\rm Gamma}(\alpha_j,\lambda)$, for each $j=1,\ldots,m$, from the identity $\mathbb{E}[\log(Z)]=\psi(a)-\log(b)$  for $Z\sim{\rm Gamma}(a,b)$, the proofs readily follows.
\end{proof}

By combining Proposition \ref{proposition:theil-tl} and Theorem \ref{theorem:expected-value-tl} the following result follows.

\begin{corollary} 
	The bias of \(\widehat{T}_L\) relative to \(T_L\), denoted by \(\mathrm{Bias}(\widehat{T}_L, T_L)\), can be expressed as
	\[
	\mathrm{Bias}(\widehat{T}_L, T_L)
	=
	\sum_{1 \leqslant j_1, \dots, j_n \leqslant m}
	\pi_{j_1} \cdots \pi_{j_n} \, \psi\left( \sum_{i=1}^n \alpha_{j_i} \right) - \log(n) - \log\left( \sum_{j=1}^m \pi_j \alpha_j \right),
	\]
	where \(\psi(\cdot)\) is the digamma function.
\end{corollary}

{
	\begin{remark}
		Observe that when \(\alpha_j = \alpha \text{ for all } j=1,\ldots,m\), the resulting bias becomes:
		\begin{align*}
			\mathrm{Bias}(\widehat{T}_L, T_L)
			=
			\psi\left(n \alpha\right)-\log(n\alpha)<0.
		\end{align*}
		Thus, Proposition 3.11 of \cite{VilaSaulo2025TheilAtkinsonGamma} is verified.
	\end{remark}
}

\subsection{Bias of the Atkinson index estimator}\label{Bias of the Atkinson index estimator}

{
	The Atkinson index estimator is defined as:
	\begin{align*}
		\widehat{A(\varepsilon)}
		=
		1-
		\dfrac{{\dfrac{1}{n} \displaystyle  \left(\sum_{i=1}^n X_i^{1-\varepsilon}\right)^{1/(1-\varepsilon)}}}
		{\dfrac{1}{n} \displaystyle \sum_{i=1}^n X_i}, \quad
		0\leqslant \varepsilon\neq 1, \
		n \in \mathbb{N},
	\end{align*}
	where $X_1,\ldots, X_n$ are i.i.d. observations from the population $X$.

	In this subsection we are interested in calculating exactly the expectation of $\widehat{A(\varepsilon)}$ for the cases $\varepsilon\to 1$ and $\varepsilon\to \infty$ where \(X_1, \ldots, X_n\) are independent copies of \(X \sim \mathrm{GM}(\bm{\theta})\).

	\begin{definition}
		The generalized mean (also called the power mean) of a set of positive real numbers $x_1,\ldots,x_n$ with exponent
		$p\in\mathbbm{R}\backslash\{0\}$ is defined as:
		\begin{align*}
			M_p(x_1,\ldots,x_n)
			\equiv
			\left({\dfrac{1}{n}}\sum _{i=1}^{n}x_{i}^{p}\right)^{{1}/{p}}.
		\end{align*}
	\end{definition}
	Using the definition and well-known limits \citep{Bullen2003}:
	\begin{align*}
		\lim_{p\to p_0}
		\displaystyle M_{p}(x_{1},\ldots,x_{n})
		=
		\begin{cases}
			\sqrt[n]{x_1\cdots x_n},
			& \text{if} \ p_0=0,
			\\[0,3cm]
			\min\{x_1,\ldots,x_n\}, & \text{if} \ p_0=-\infty,
		\end{cases}
	\end{align*}
	we have
}
%
\[
\widehat{A(1)}
\equiv \lim_{\varepsilon\to 1}\widehat{A(\varepsilon)}
=
1 - \dfrac{\displaystyle \left(\prod_{i=1}^n X_i\right)^{1/n}}{\dfrac{1}{n} \displaystyle \sum_{i=1}^n X_i}, \quad n \in \mathbb{N},
\]
and
\[
\widehat{A(\infty)}
\equiv \lim_{\varepsilon\to \infty}\widehat{A(\varepsilon)}
=
1 -
\dfrac{\min\{X_1, \dots, X_n\}}{\dfrac{1}{n} \displaystyle \sum_{i=1}^n X_i}, \quad n \in \mathbb{N}.
\]

\begin{theorem}
	\label{Expected Value of the Atkinson estimator-1}
	If \(X_1, \dots, X_n\) are independent copies of \(X \sim \text{GM}(\bm{\theta})\), then
	\[
	\mathbb{E}\left[\widehat{A(1)}\right]
	=
	1
	-
	n
	\sum_{1 \leqslant j_1, \dots, j_n \leqslant m}
	\pi_{j_1} \cdots \pi_{j_n} \,
	\dfrac{\displaystyle 1}{\sum_{i=1}^n \alpha_{j_i}} \,
	\prod_{i=1}^n
	\dfrac{\Gamma(\alpha_{j_i} + {1\over n})}{\Gamma(\alpha_{j_i})}.
	\]
\end{theorem}

\begin{proof}
	Note that \(\widehat{A(1)}\) can be written as
	\[
	\widehat{A(1)} = 1 - n \left( \prod_{i=1}^n D_i \right)^{1/n},
	\]
	{where  $D_i = {X_i/S_n}$ and
		$S_n = \sum_{i=1}^n X_i.$} By \eqref{cond-dirichlet},
	given \(Y_1=j_1, \ldots, Y_n=j_n\) the vector $(D_1, \ldots, D_n)^\top$ has Dirichlet distribution with parameter vector $(\alpha_{j_1}, \dots, \alpha_{j_n})^\top$.
	Hence, by applying Identity \eqref{identity:expected-value-prod-pi}, with \(d_i = 1/n\), we have
	\begin{align*}
		\mathbb{E}\left[\widehat{A(1)} \, \Big\vert \, Y_1=j_1, \ldots, Y_n=j_n\right]
		&=
		1 - n \, \mathbb{E} \left[ \left( \prod_{i=1}^n D_i \right)^{1/n} \Bigg| \,  Y_1=j_1, \ldots, Y_n=j_n \right]
		\\[0,2cm]
		&= 1 - n \, \frac{\Gamma\left(\sum_{i=1}^n \alpha_{j_i}\right)}{\Gamma\left(\sum_{i=1}^n \alpha_{j_i} + 1\right)} \prod_{i=1}^n \frac{\Gamma(\alpha_{j_i} + {1\over n})}{\Gamma(\alpha_{j_i})}
		\\[0,2cm]
		&= 1 - \frac{n}{\sum_{i=1}^n \alpha_{j_i}} \prod_{i=1}^n \frac{\Gamma(\alpha_{j_i} + {1\over n})}{\Gamma(\alpha_{j_i})},
	\end{align*}
	where the last line follows from the well-known identity: $\Gamma(x+1)=x\Gamma(x)$.
	Finally, applying the Law of total probability, the stated result follows immediately.
\end{proof}

{
	From Theorem \ref{Expected Value of the Atkinson estimator-1} and Proposition \ref{pro-atkinson-infty}, it follows that:
	\begin{corollary}
		The bias of \(\widehat{A(1)} \) relative to \({A(1)} \), denoted by \(\mathrm{Bias}(\widehat{A(1)}, {A(1)} )\), is given by
		\begin{multline*}
			\mathrm{Bias}(\widehat{A(1)}, {A(1)} )
			=
			-
			n
			\sum_{1 \leqslant j_1, \dots, j_n \leqslant m}
			\pi_{j_1} \cdots \pi_{j_n} \,
			\dfrac{\displaystyle 1}{\sum_{i=1}^n \alpha_{j_i}}
						\\[0,2cm]
						\times
			\prod_{i=1}^n
			\dfrac{\Gamma(\alpha_{j_i} + {1\over n})}{\Gamma(\alpha_{j_i})}
			+
			\frac{1}{\sum_{j=1}^m \pi_j \alpha_j} \exp\left\{\sum_{j=1}^m \pi_j \psi(\alpha_j)\right\}.
		\end{multline*}
	\end{corollary}

	\begin{remark}
		When \(\alpha_j = \alpha \text{ for all } j=1,\ldots,m\), the resulting bias becomes:
		\begin{align*}
			\mathrm{Bias}(\widehat{A(1)}, {A(1)} )
			=
			\dfrac{\displaystyle 1}{\alpha}
			\left[
			\exp\left\{\psi(\alpha)\right\}
			-
			\dfrac{\Gamma^n(\alpha + {1\over n})}{\Gamma^n(\alpha)}
			\right],
		\end{align*}
		which confirms the validity of recent result in Corollary 3.17 of \cite{VilaSaulo2025TheilAtkinsonGamma}.
	\end{remark}

}

\begin{theorem}
	\label{Expected Value of the Atkinson estimator}
	If \(X_1, \dots, X_n\) are independent copies of \(X \sim \text{GM}(\bm{\theta})\), then
	{
		\begin{multline*}
			\mathbb{E}\left[\widehat{A(\infty)}\right]
			=
			1
			-
			n
			\sum_{\substack{k_1+\cdots+k_m=n \\[0,1cm] k_1,\ldots,k_m\geqslant 0}}
			\binom{n}{k_1,\ldots,k_m}
			\pi_1^{k_1}\cdots \pi_m^{k_m}
			\,
			\dfrac{1}{\sum_{j=1}^m\alpha_j k_j}
			\\[0,2cm]
			\times
			\int_0^\infty
			\dfrac{
				\Gamma^{k_1}(\alpha_1,u)}{\Gamma^{k_1}(\alpha_1)}
			\cdots
			\dfrac{
				\Gamma^{k_m}(\alpha_m,u)
			}{\Gamma^{k_m}(\alpha_m)} \,
			{\rm d}u,
		\end{multline*}
	}
	where
	\begin{align*}
		{\displaystyle {n \choose k_{1},\ldots ,k_{m}}={\frac {n!}{k_{1}!\cdots k_{m}!}}}
	\end{align*}
	is a multinomial coefficient. The sum is over all nonnegative integers $k_1,\ldots,k_m$ with total sum $n$.
	%
\end{theorem}
\begin{proof}
	{
		By using the identity
		\begin{align*}
			\int_{0}^\infty \exp(-\xi z){\rm d}z={1\over \xi},
			\quad \xi>0,
		\end{align*}
		with $\xi=\sum_{i=1}^{n}X_i$, we get
		\begin{align}\label{min-formula-0}
			\mathbb{E}
			\left[
			\dfrac{\min\{X_1, \dots, X_n\}}{\displaystyle \sum_{i=1}^n X_i}
			\right]
			=
			\mathbb{E}\left[\min\{X_{1},\ldots,X_{n}\}\int_0^\infty\exp\left\{-\left(\sum_{i=1}^{n}X_i\right)z\right\}{\rm d}z\right].
		\end{align}
		By applying the identity
		\begin{align*}
			\min\{X_{1},\ldots,X_{n}\}
			=
			\int_0^\infty
			\mathds{1}_{\bigcap_{i=1}^{n} \{X_i\geqslant t\}}
			{\rm d}t,
		\end{align*}
		the expression in \eqref{min-formula-0} becomes
		\begin{align*}
			\int_0^\infty
			\int_0^\infty
			\mathbb{E}\left[
			\mathds{1}_{\bigcap_{i=1}^{n} \{X_i\geqslant t\}}
			\exp\left\{-\left(\sum_{i=1}^{n}X_i\right)z\right\}\right]
			{\rm d}t
			{\rm d}z,
		\end{align*}
		where the interchange of integrals is justified by Tonelli's theorem. Since $X_1,\ldots, X_n$  are i.i.d. with the same distribution of $X \sim \mathrm{GM}(\bm{\theta})$, the last expression simplifies to
		\begin{align}\label{integraldouble}
			\int_0^\infty
			\int_0^\infty
			\mathbb{E}^n\left[
			\mathds{1}_{\{X\geqslant t\}}
			\exp\left(-Xz\right)\right]
			{\rm d}t
			{\rm d}z.
		\end{align}
		When $X\sim\text{GM}(\bm{\theta})$, direct computation shows
		\begin{align*}
			\mathbb{E}\left[
			\mathds{1}_{\{X\geqslant t\}}
			\exp\left(-Xz\right)
			\right]
			=
			\sum_{j=1}^m
			\pi_j \,
			{\lambda^{\alpha_j} \Gamma(\alpha_j,(z+\lambda)t)\over (z+\lambda)^{\alpha_j} \Gamma(\alpha_j)}.
		\end{align*}
		Consequently, by using the above identity and then the Multinomial theorem, the integral in \eqref{integraldouble} can be written as
		\begin{align*}
			&\int_0^\infty
			\int_0^\infty
			\left[
			\sum_{j=1}^m
			\pi_j \,
			{\lambda^{\alpha_j}\over (z+\lambda)^{\alpha_j} \Gamma(\alpha_j)}
			\,
			\Gamma(\alpha_j,(z+\lambda)t)
			\right]^n
			{\rm d}t
			{\rm d}z
			\\[0,2cm]
			&=
			\sum_{\substack{k_1+\cdots+k_m=n \\[0,1cm] k_1,\ldots,k_m\geqslant 0}}
			\binom{n}{k_1,\ldots,k_m}
			\prod_{j=1}^m \pi_j^{k_j}
			\int_0^\infty
			\int_0^\infty
			\prod_{j=1}^m
			\dfrac{\lambda^{\alpha_j k_j}\Gamma^{k_j}(\alpha_j,(z+\lambda)t)}{(z+\lambda)^{\alpha_j k_j} \Gamma^{k_j}(\alpha_j)}
			{\rm d}t
			{\rm d}z.
		\end{align*}
		Now, making the change of variable $u=(z+\lambda)t$, the proof of theorem follows.
	}
	%
	%
	%
	%
	%
\end{proof}

{
	By combining Theorem \ref{Expected Value of the Atkinson estimator} and Proposition \ref{pro-atkinson-1} we have:
	\begin{corollary}\label{bias-At}
		The bias of \(\widehat{A(\infty)} \) relative to \({A(\infty)} \), denoted by \(\mathrm{Bias}(\widehat{A(\infty)}, {A(\infty)} )\), is given by
		\begin{multline*}
			\mathrm{Bias}(\widehat{A(\infty)}, {A(\infty)} )
			=
			-
			n
			\sum_{\substack{k_1+\cdots+k_m=n \\[0,1cm] k_1,\ldots,k_m\geqslant 0}}
			\binom{n}{k_1,\ldots,k_m}
			\pi_1^{k_1}\cdots \pi_m^{k_m}
			\,
			\dfrac{1}{\sum_{j=1}^m\alpha_j k_j}
			\\[0,2cm]
			\times
			\int_0^\infty
			\dfrac{
				\Gamma^{k_1}(\alpha_1,u)}{\Gamma^{k_1}(\alpha_1)}
			\cdots
			\dfrac{
				\Gamma^{k_m}(\alpha_m,u)
			}{\Gamma^{k_m}(\alpha_m)} \,
			{\rm d}u.
		\end{multline*}
	\end{corollary}

	\begin{remark}
		Note that the integrals in Theorem \ref{Expected Value of the Atkinson estimator} and Corollary \ref{bias-At} do not have closed-form expressions in terms of standard mathematical functions, and therefore require numerical methods for evaluation.
	\end{remark}

	\begin{remark}
		The scale invariance property of estimators $\widehat{T}_T$, $\widehat{T}_L$, $\widehat{A(1)}$ and $\widehat{A(\infty)}$ implies that their respective expectations are unaffected by the rate $\lambda$, as shown in Theorems \ref{theorem:expected-value-tt}, \ref {theorem:expected-value-tl}, \ref{Expected Value of the Atkinson estimator-1} and \ref{Expected Value of the Atkinson estimator}, respectively.
	\end{remark}

}

\subsection{Bias of the  dispersion index estimator}\label{Bias of the  dispersion index estimator}

The sampling estimator of the variance-to-mean ratio (dispersion  index) is given by

\begin{equation*}
	\widehat{\mathrm{VMR}}
	\equiv
	\frac{S^2}{\overline{X}}
	=
	\frac{1}{n - 1}
	\sum_{i=1}^{n} \frac{(X_i - \overline{X})^2}{\overline{X}}, \quad n \geqslant 2,
\end{equation*}
{
	where $X_1,\ldots, X_n$ denote i.i.d. samples drawn from the population  $X$, $\overline{X}$ is the sample mean and $S^2$ is the sample variance.
}

\begin{theorem}
	\label{theorem:ev-vmr}
	If \(X_1, \dots, X_n\) are independent copies of \(X \sim \text{GM}(\bm{\theta})\), then
	\begin{align*}
	\mathbb{E}\left[\widehat{\mathrm{VMR}}\right]
		=
	\frac{n}{(n-1)\lambda}
	\sum_{1\leqslant j_1, \dots, j_n \leqslant m}
	\pi_{j_1} \cdots \pi_{j_n}
	\left[
	\dfrac{1}{\sum_{i=1}^n \alpha_{j_i} + 1}
	\sum_{i=1}^{n} {\alpha_{j_i}(\alpha_{j_i} + 1)}
	-
	{1\over n}
	{\sum_{i=1}^n \alpha_{j_i}}
	\right].
	\end{align*}
\end{theorem}

\begin{proof}
	Notice that, using the notations of Remark~\ref{remark:rewrite-tt}, the estimator $\widehat{\mathrm{VMR}}$  can be written as
	\[
	\widehat{\mathrm{VMR}}
	=
	\frac{n}{n-1} \,
	S_n
	\sum_{i=1}^{n} \left(D_i - \frac{1}{n} \right)^2
	=
	\frac{n}{n-1} \,
	S_n
	\left(\sum_{i=1}^{n} D_i^2 - \frac{1}{n} \right).
	\\[0.2em]
	\]

	Given \(Y_1=j_1,\ldots,Y_n=j_n\), \(S_n\) and \((D_1, \ldots, D_{n})^\top\) are independent (Theorem \ref{Mosimann}). Hence,
	\begin{align*}
		\mathbb{E}\left[\widehat{\mathrm{VMR}} \, \Big\vert \, Y_1=j_1,\ldots,Y_n=j_n \right]
		& =
		\frac{n}{n-1}  \,
		\mathbb{E}\left[S_n \mid Y_1=j_1, \ldots, Y_n=j_n \right] \\[0,2cm]
		&\times \mathbb{E}\left[ \sum_{i=1}^{n} D_i^2 - \frac{1}{n} \,\Bigg \vert \, Y_1=j_1,\ldots,Y_n=j_n \right].
	\end{align*}

	As \( S_n \mid Y_1=j_1,\dots,Y_n=j_n \sim \text{Gamma}(\sum_{i=1}^n \alpha_{j_i}, \lambda)\), we have \(\mathbb{E}[S_n \mid Y_1=j_1,\dots,Y_n=j_n ] = \sum_{i=1}^n \alpha_{j_i} / \lambda\). Consequently, by using Proposition~\ref{proposition:moments-dirichlet-0}
	with $r=1$ and $c_j=2\delta_{ji}$,
	%
	we get
	\begin{multline*}
	\mathbb{E}\left[\widehat{\mathrm{VMR}} \, \Big\vert\, Y_1=j_1,\ldots,Y_n=j_n \right]
	\\[0,2cm]
	=
	\frac{n}{n-1} \,
	\frac{\sum_{i=1}^n \alpha_{j_i}}{\lambda}
	\left[
	\dfrac{1}{(\sum_{i=1}^n \alpha_{j_i})(\sum_{i=1}^n \alpha_{j_i} + 1)}
	\sum_{i=1}^{n} {\alpha_{j_i}(\alpha_{j_i} + 1)} - \frac{1}{n}
	\right].
		\end{multline*}
	By applying the Law of total expectation, the result follows.
\end{proof}

As a direct result of Corollary~\ref{corollary:vmr} and Theorem~\ref{theorem:ev-vmr}, we obtain the following result.
\begin{corollary}
	The bias of \(\widehat{\mathrm{VMR}}\) relative to \(\mathrm{VMR}\), denoted by \(\mathrm{Bias}(\widehat{\mathrm{VMR}}, \mathrm{VMR})\), is given by
	\begin{multline*}
		\mathrm{Bias}(\widehat{\mathrm{VMR}}, \mathrm{VMR})
		=
		\: \frac{n}{(n - 1)\lambda}
		\sum_{1\leqslant j_1, \dots, j_n \leqslant m}
		\pi_{j_1} \cdots \pi_{j_n}
		\left[
		\dfrac{1}{\sum_{i=1}^{n} \alpha_{j_i} + 1 }
		{ \sum_{i=1}^{n} \alpha_{j_i} (\alpha_{j_i} + 1) }
		- \frac{1}{n} \sum_{i=1}^{n} \alpha_{j_i}
		\right]
		\\[0,2cm]
		-
		\dfrac{1}{\lambda \sum_{j=1}^{m} \pi_j \alpha_j}
		\left[
		{\sum_{j=1}^{m} \pi_j \alpha_j (\alpha_j + 1) - \left( \sum_{j=1}^{m} \pi_j \alpha_j \right)^2}
		\right].
	\end{multline*}
\end{corollary}

\begin{remark}
	Note that when \(\alpha_j = \alpha \text{ for all } j=1,\ldots,m\), the resulting bias is:
	\[
	\mathrm{Bias}(\widehat{\mathrm{VMR}}, \mathrm{VMR})
	=
	\frac{n\alpha}{\lambda(n\alpha + 1)}
	-
	\frac{1}{\lambda}
	=
	- \frac{1}{\lambda(n\alpha + 1)} < 0,
	\]
	which corroborates the recent result reported in \cite{Shih2025}.
\end{remark}

In Table \ref{table:0-1}, we summarize all the biases of the estimators obtained in Subsections \ref{subsec:expected-value-theil}, \ref{Bias of the Atkinson index estimator} and \ref{Bias of the  dispersion index estimator}.
\begin{table}[H]
	\centering
	\caption{Theil, Atkinson and dispersion estimators and their respective biases.}
	\label{table:0-1}
	\hspace*{-1cm}
	\resizebox{\linewidth}{!}
	{
		\begin{tabular}[t]{cc}
			\toprule
			Estimator &  Bias  \\
			\midrule
			\rowcolor{gray!10}
			$\widehat{T}_T$
			&
			\makecell{
				$\displaystyle
				\sum_{1 \leqslant j_1, \dots, j_n \leqslant m}
				\pi_{j_1} \cdots \pi_{j_n}
				\,
				\frac{1}{\sum_{i=1}^n \alpha_{j_i}}
				\left[
				\sum_{i=1}^n \alpha_{j_i} \psi(\alpha_{j_i})
				-
				\left( \sum_{i=1}^n \alpha_{j_i} \right) \psi\left( \sum_{i=1}^n \alpha_{j_i} \right)
				+
				n-1
				\right]  \hspace{0.5cm}$
				\\[0,8cm]
				$\displaystyle
				+
				\log(n)
				- \left[
				\frac{1}{\sum_{j=1}^{m} \pi_j \alpha_j} \left(
				\sum_{j=1}^{m} \pi_j \alpha_j \psi(\alpha_j)
				+ 1 \right)
				- \log\left(\sum_{j=1}^{m} \pi_j \alpha_j \right) \right]
				$}
			\\[3.0ex] \addlinespace
			$\widehat{T}_L$
			&
			$\displaystyle
			\sum_{1 \leqslant j_1, \dots, j_n \leqslant m}
			\pi_{j_1} \cdots \pi_{j_n} \, \psi\left( \sum_{i=1}^n \alpha_{j_i} \right) - \log(n) - \log\left( \sum_{j=1}^m \pi_j \alpha_j \right)
			$
			\\[3.0ex] \addlinespace \rowcolor{gray!10}
			$\widehat{A(1)}$
			&
			$ \displaystyle
			-
			n
			\sum_{1 \leqslant j_1, \dots, j_n \leqslant m}
			\pi_{j_1} \cdots \pi_{j_n} \,
			\dfrac{\displaystyle 1}{\sum_{i=1}^n \alpha_{j_i}} \,
			\prod_{i=1}^n
			\dfrac{\Gamma(\alpha_{j_i} + {1\over n})}{\Gamma(\alpha_{j_i})}
			+
			\frac{1}{\sum_{j=1}^m \pi_j \alpha_j} \exp\left\{\sum_{j=1}^m \pi_j \psi(\alpha_j)\right\}
			$
			\\[3.0ex]  \addlinespace
			$\widehat{A(\infty)}$
			&
			$ \displaystyle
			-
			n
			\sum_{\substack{k_1+\cdots+k_m=n \\[0,1cm] k_1,\ldots,k_m\geqslant 0}}
			\binom{n}{k_1,\ldots,k_m}
			\pi_1^{k_1}\cdots \pi_m^{k_m}
			\,
			\dfrac{1}{\sum_{j=1}^m\alpha_j k_j}
			\int_0^\infty
			\dfrac{
				\Gamma^{k_1}(\alpha_1,u)}{\Gamma^{k_1}(\alpha_1)}
			\cdots
			\dfrac{
				\Gamma^{k_m}(\alpha_m,u)
			}{\Gamma^{k_m}(\alpha_m)} \,
			{\rm d}u
			$
			\\[3.0ex] \addlinespace \rowcolor{gray!10}
			$\widehat{\mathrm{VMR}}$
			&
			\makecell{
				$\displaystyle
				\frac{n}{(n - 1)\lambda}
				\sum_{1\leqslant j_1, \dots, j_n \leqslant m}
				\pi_{j_1} \cdots \pi_{j_n}
				\left[
				\dfrac{1}{\sum_{i=1}^{n} \alpha_{j_i} + 1 }
				{ \sum_{i=1}^{n} \alpha_{j_i} (\alpha_{j_i} + 1) }
				- \frac{1}{n} \sum_{i=1}^{n} \alpha_{j_i}
				\right] \hspace{1.7cm} $
				\\[0,8cm]
				$\displaystyle
				-
				\dfrac{1}{\lambda \sum_{j=1}^{m} \pi_j \alpha_j}
				\left[
				{\sum_{j=1}^{m} \pi_j \alpha_j (\alpha_j + 1) - \left( \sum_{j=1}^{m} \pi_j \alpha_j \right)^2}
				\right]
				$
			}
			\\[3.0ex]
			\bottomrule
		\end{tabular}
	}
	\hspace*{-1cm}
\end{table}

\section{Illustrative simulation study}\label{cap:simulacao}

In this section, a Monte Carlo simulation study is conducted to empirically assess the performance of the estimators under different sample sizes and population settings. We compare the traditional and bias-corrected versions of the estimators of the inequality indices (Theil-T, Theil-L, Atkinson with $\varepsilon = 1$, Atkinson in the limit $\varepsilon \to \infty$, and the dispersion index VMR), examining how the analytical corrections affect the estimators in terms of bias and mean squared error (MSE). Particular attention is given to scenarios involving small and moderate samples, where bias tends to be more pronounced.

A two-component mixture was considered with fixed weights $(\pi_1, \pi_2) = (0.60; 0.40)$, shape parameter $\alpha_1 = 0.5$ for the first component, and $\alpha_2 \in \{0.5; 1; 2; 3; 4; 5\}$ for the second, keeping the rate parameter fixed at $\lambda = 1$. This configuration ranges from a homogeneous case ($\alpha_2 = 0.5$) to more concentrated distributions ($\alpha_2 = 5.0$). Sample sizes varied from $n = 10$ to $n = 100$, and the number of replications was fixed at $N_{\text{rep}} = 1{,}000$. Observations were generated according to Definition~\ref{definition:mixture-gamas}, where each $X_i$ comes from the $k$-th component with probability $\pi_k$. For each sample, the mixture parameters were estimated by maximum likelihood, yielding $(\widehat{\pi}_1, \widehat{\pi}_2, \widehat{\alpha}_1, \widehat{\alpha}_2, \widehat{\lambda})$, and the indices and their corrected versions were computed. In all cases, the bias-corrected estimator associated with an index $\theta$
is defined by
\begin{equation}
    \widehat{\theta}^{\,bc}
    \;=\;
    \widehat{\theta}
    \;-\;
    \widehat{\mathrm{Bias}}(\theta),
\end{equation}
where $\widehat{\theta}$ denotes the traditional (plug-in) estimator
and $\widehat{\mathrm{Bias}}(\theta)$ is the analytical bias from
Table~\ref{table:0-1}, evaluated at the estimated mixture parameters
$(\widehat{\pi}_j,\widehat{\alpha}_j,\widehat{\lambda})$.

The entire procedure—from sample generation, parameter estimation, index computation, application of the corrections, and final evaluation—is summarized in Algorithm~\ref{alg:simulacao}, which is run identically for all indices considered, varying only the functional $\theta$ of interest.

\vspace{6pt}
\begin{algorithm}
    \caption{Monte Carlo simulation for bias-corrected estimators.}
    \label{alg:simulacao}

    \KwIn{
        Number of replications $N_{\text{rep}} = 1{,}000$;\\
        Sample sizes $n \in \{10, 11, \ldots, 100\}$;\\
        True parameters $(\pi_1, \pi_2) = (0.6; 0.4)$, $\alpha_1 = 0.5$,
        $\alpha_2 \in \{0.5, 1, 2, 3, 5\}$ and $\lambda = 1$.
    }

    \KwOut{
        Bias, MSE and Monte Carlo standard errors of the traditional and corrected estimators.
    }

    \BlankLine

    \ForEach{combination of $(n, \alpha_2)$}{
        \For{$r \gets 1$ \KwTo $N_{\text{rep}}$}{
            Generate $X_1^{(r)}, \ldots, X_n^{(r)}$ from the Gamma mixture\;
            Estimate the mixture parameters via maximum likelihood\;

            \ForEach{$\theta \in \{\text{Theil-}T, \text{Theil-}L, A(1), A(\infty), \text{VMR}\}$}{
                Compute the traditional estimator $\widehat{\theta}^{(r)}$\;
                Compute the correction term $\widehat{B}^{(r)}(\theta)$\;
                Compute the corrected estimator
                $\widehat{\theta}_{bc}^{(r)} = \widehat{\theta}^{(r)} - \widehat{B}^{(r)}(\theta)$\;
            }
        }

        Compute bias and MSE\;

        Compute the Monte Carlo standard errors\;
    }

    \Return{
        The estimates of bias, MSE and Monte Carlo standard errors for
        $\widehat{\theta}$ and $\widehat{\theta}_{bc}$.
    }
\end{algorithm}

\subsection{Estimation of the gamma mixture parameters}

The simulation requires, in each replication, the estimation of the parameters of the gamma mixture distribution
$\mathrm{MG}(\pi,\alpha,\lambda)$ from the samples generated according to
Definition~\ref{definition:mixture-gamas}. Estimation was carried out by maximum likelihood, via direct numerical optimization of the log-likelihood function of the mixture.

The parametrization adopted imposes, in a continuous and differentiable way, all restrictions required by the model specification, thereby avoiding common numerical difficulties in mixture estimation. To ensure that $\pi_1 \in (0,1)$, the logistic transformation
$\pi_1 = \exp(\eta)/(1+\exp(\eta))$ was used. Positivity of the shape and rate parameters was enforced through box constraints in the optimization algorithm. In addition, the ordering convention is followed, in which the shape parameters are reparametrized in terms of $(a,\delta)$, with $a>0$ and $\delta \in \mathbb{R}$. This formulation guarantees a consistent ordering of the components throughout the optimization, avoiding numerically unstable solutions or the need for posterior reordering.

Optimization was performed using the L-BFGS-B method (\emph{Limited-memory
	Broyden--Fletcher--Goldfarb--Shanno with Bounds}) via the \texttt{optim} function in
\textsf{R}, specifying positive lower bounds for $a$ and $\lambda$. Initial values were defined from a pre-specified vector on the reparametrized scale, and convergence was assessed using the standard criterion of the method, based on the relative stabilization of the likelihood function. For each replication, the final solution corresponds to the local maximum obtained from this starting point.

\subsection{Computational implementation difficulties}

During the development of the routines for computing estimator biases, implemented in the \texttt{R} environment (version~4.4.2) \citep{RCore2024}, it was observed that the main practical limitation lay in the computational cost associated with summations over all possible compositions of the sample size. These expressions involve high-order combinatorial terms whose cardinality grows exponentially with the number of components and with the sample size $n$. As a consequence, direct evaluation of the theoretical formulas quickly becomes infeasible, especially in scenarios with greater heterogeneity and when one seeks to estimate the bias for several parameter configurations.

Moreover, the use of factorials in expressions involving combinatorial coefficients such as $\binom{n}{k_1, k_2, \ldots, k_m}$ led to numerical overflow in operations involving factorials, since the values of $n!$ grow superexponentially. This behavior makes direct computation impossible even for moderately large samples, producing undefined or inaccurate results under standard floating-point arithmetic in \texttt{R}.

To circumvent these difficulties, a twofold optimization strategy was adopted. First, the summations were reorganized so as to group terms with identical combinatorial weights, reducing redundancy and the effective number of operations. To illustrate, in the scenario studied with only two components, it is possible to show that the multinomial sum can be rewritten as
\[
\sum_{\substack{k_1+\cdots+k_m=n \\[0.1cm] k_1,\ldots,k_m\geqslant 0}}
        \binom{n}{k_1,\ldots,k_m}
        \pi_1^{k_1}\cdots \pi_m^{k_m}
        \;\overset{m=2}{=}\;
        \sum_{k=0}^n \binom{n}{k}\,\pi_1^{k}(1-\pi_1)^{n-k}.
\]
In the general case, the number of terms in the sum on the left-hand side is given by the number of integer compositions of $n$ into $m$ parts, that is,
\[
N_{\text{comp}} = \binom{n+m-1}{m-1},
\]
which implies a computational cost of order $\mathcal{O}\!\left(\binom{n+m-1}{m-1}\right)$.
However, by exploiting symmetries among the compositions and grouping terms with identical combinatorial weights, the number of operations is reduced to $\mathcal{O}(n)$ when $m=2$.
This simplification yields a reduction of one to two orders of magnitude in execution time, without any loss of analytical exactness.

The second optimization strategy concerns replacing the direct computation of factorials with the log-gamma function, \texttt{lgamma}, which provides a numerically stable way to compute $\log\big(\Gamma(x)\big)$ without overflow. This reformulation preserves arithmetic accuracy and makes it possible to work with much larger values of $n$, transforming multiplications and divisions of large numbers into additions and subtractions in the logarithmic scale.

These modifications were decisive for making the algorithm executable in reasonable time, enabling bias estimation under multiple parameter configurations and sample sizes. In summary, the reformulation of the sums and the use of logarithmic representations constitute a compromise between theoretical exactness and numerical efficiency, a key aspect for the success of the implementation.

\subsection{Simulation results}

The simulation results are presented in Figures~\ref{fig:fig_indices_alpha} and~\ref{fig:fig_indices_n}, which show the average behavior of bias and MSE of the estimators as functions of $\alpha_2$ and $n$, respectively. Table~\ref{ap:mcse} complements these figures by listing, for each parameter combination, the values of $\theta$ and $\hat{\theta}$, the empirical bias, the MSE and the standard errors associated with these estimates, thereby allowing assessment of the numerical accuracy of the simulations and confirming the general trends observed in the graphs.

Figure~\ref{fig:fig_indices_alpha} shows the average bias (left column) and MSE (right column) of the estimators as a function of $\alpha_2$, keeping $\alpha_1 = 0.5$ and $n = 15$ fixed. For the Theil-T index, the traditional estimator (dashed line) exhibits a substantial negative bias for small $\alpha_2$, which decreases as heterogeneity increases. The corrected version (solid line) practically eliminates this bias, remaining close to zero across the entire range, a feature also confirmed by the tabulated values in Table~\ref{ap:mcse}, whose Monte Carlo standard errors remain small for all parameter combinations. The same behavior is observed for the Theil-L index, where the correction substantially reduces the magnitude of the bias and stabilizes the estimates. Similar results are found for the Atkinson index with $\varepsilon = 1$, whose bias is reduced and becomes practically null after correction. In the limit case $\varepsilon \to \infty$, a slight sign reversal (“supercorrection”) is observed for small $\alpha_2$. This effect is expected because $A(\infty)$ depends strongly on the sample minimum, which exhibits high variability when $n$ is small; in such situations, the analytical correction based on population expected values may temporarily overcompensate the bias, a phenomenon that disappears as $\alpha_2$ increases or as sample size grows. Finally, for the dispersion index (VMR), the correction significantly reduces the negative bias, although some variability persists as $\alpha_2$ increases, a pattern likewise reflected in the values presented in Table~\ref{ap:mcse}.

Figure~\ref{fig:fig_indices_n} shows the average bias (left column) and MSE (right column) of the estimators as a function of the sample size. For the Theil and Atkinson $(\varepsilon = 1)$ indices, one initially observes a negative bias in small samples, indicating \textit{systematic underestimation} of inequality when sample sizes are small. This bias is gradually corrected as $n$ increases, becoming practically null from $n = 30$ onward. The Theil-L estimator exhibits greater variability in small samples, reflecting its sensitivity to the lower tail of the distribution. The Atkinson index $(\varepsilon \to \infty)$, which emphasizes inequality among the poorest individuals, shows bias closer to zero more markedly, suggesting lower susceptibility to sample fluctuations. Finally, the VMR exhibits a more oscillatory pattern in small samples, but converges quickly as $n$ grows. The MSE reinforces the same conclusions, showing a monotone reduction in this metric as $n$ increases, evidencing simultaneous gains in precision and stability. The difference between the traditional estimator (dashed line) and the bias-corrected estimator (solid line) is relevant only for small samples: for instance, for $n=15$, the ratio $\mathrm{MSE}_{\text{corr}} / \mathrm{MSE}_{\text{trad}}$ typically lies between $0.80$ and $0.90$ for the Theil and Atkinson indices, indicating a relative gain in precision of 10\% to 20\%. This effect becomes practically negligible in larger samples, where both estimators converge to the same asymptotic limit.

It should be noted that the bias and MSE estimates presented in the figures are subject to the inherent variability of the simulation procedure. For $N_{\text{sim}} = 1{,}000$, these standard errors are typically between $10^{-3}$ and $10^{-2}$, which are small compared to the differences observed between the traditional and corrected estimators. Therefore, the reported curves are numerically stable and do not exhibit distortions due to the finite number of simulations. The full set of standard error values, for different numbers of replications, is provided in Table~\ref{ap:mcse}, where one can verify that the qualitative conclusions remain valid even when $N_{\text{sim}}$ is increased by an order of magnitude.

In general, the Monte Carlo results indicate that all estimators are consistent and asymptotically unbiased, with increasing efficiency as the sample size grows. The inequality measures remain robust under the gamma mixture model, with satisfactory performance already in moderate samples ($n \approx 30$).

\begin{figure}[!ht]
	\centering
	\includegraphics[width=.82\textwidth]{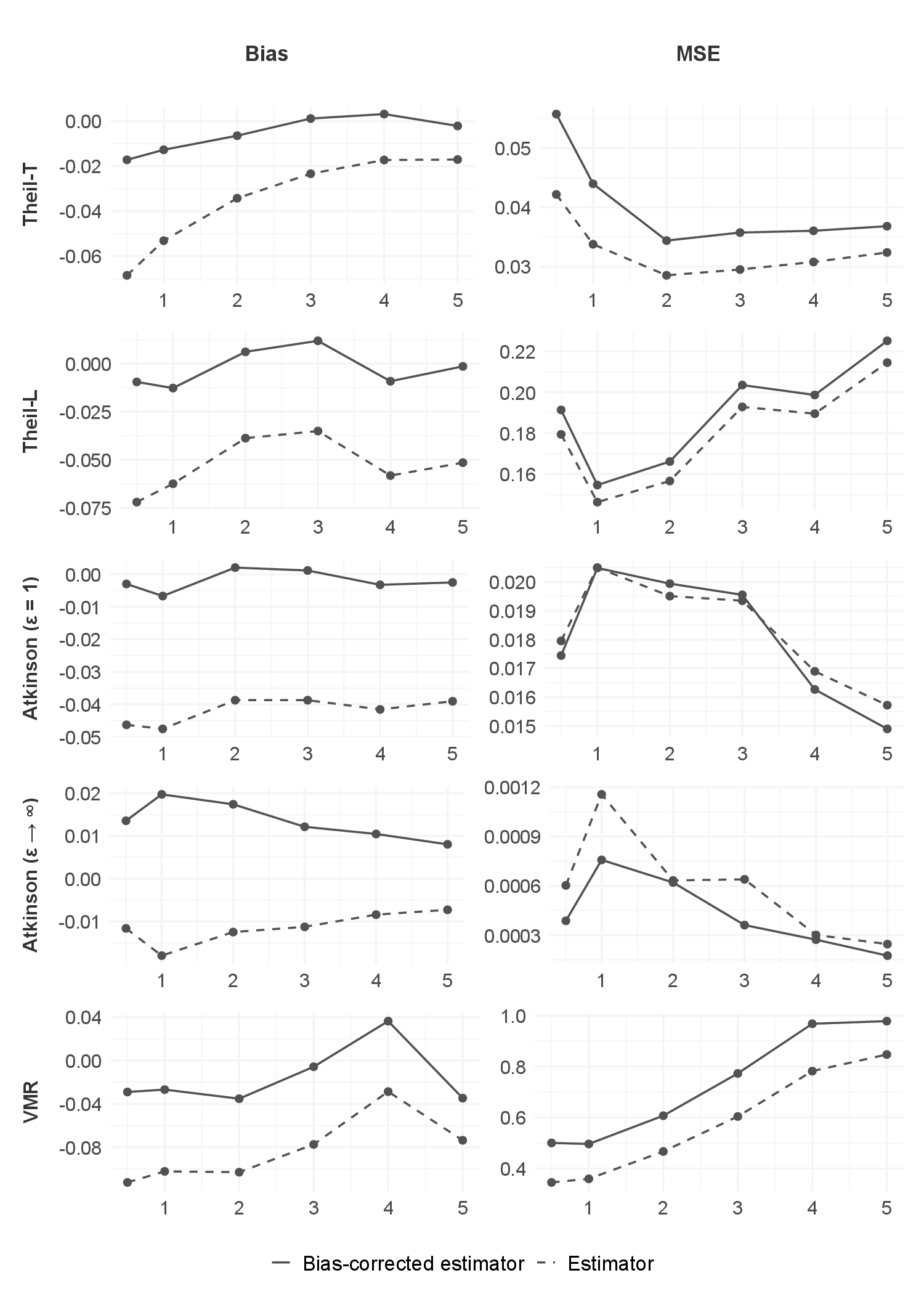}
	\caption{Bias (left) and Mean Squared Error (MSE) (right) of the indices as a function of the parameter $\alpha_2$ for $\alpha_1=0.5$ and $n=15$.}
	\label{fig:fig_indices_alpha}
\end{figure}

\begin{figure}[!ht]
	\centering
	\includegraphics[width=.82\textwidth]{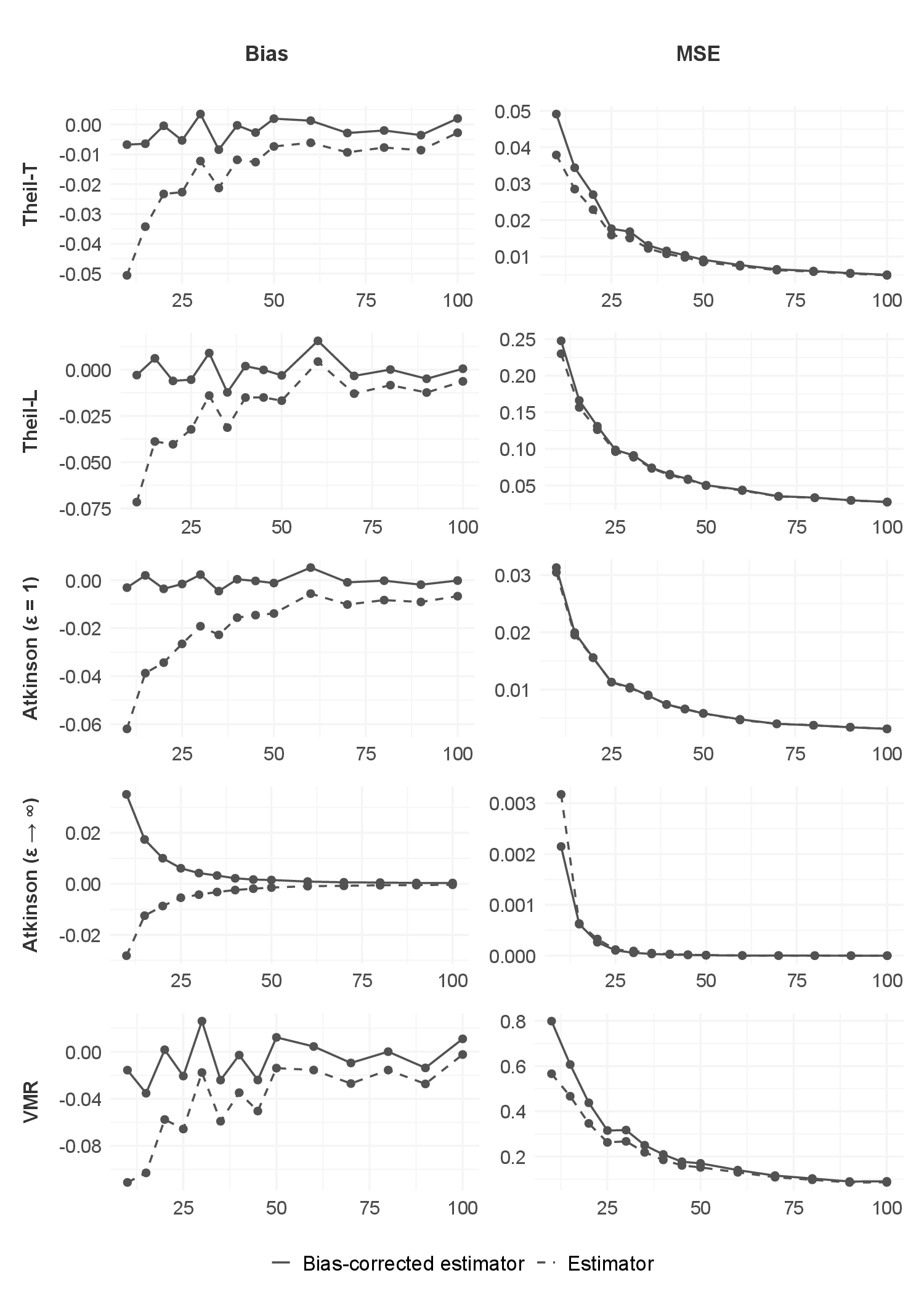}
	\caption{Bias (left) and Mean Squared Error (MSE) (right) of the indices as a function of the sample size $n$ for $\alpha_1=0.5$ and $\alpha_2={2.0}$.}
	\label{fig:fig_indices_n}
\end{figure}

{
\renewcommand{\arraystretch}{0.8}   
\setlength{\tabcolsep}{2.5pt}      
\footnotesize
\begin{longtable}{lrcrcrrrccrrrc}
\caption{{\small Results of the Monte Carlo simulations for each index and for each combination of $\alpha_2$ and $n$, showing the true value $\theta$, the mean of the estimates $\hat{\theta}$, the empirical bias $\mathrm{B}(\hat{\theta})$, the mean squared error (MSE), and the standard errors of the bias and of the MSE, both for the original estimator and for the bias-corrected estimator. The results refer to the scenario with $R = 1{,}000$ replications and $\alpha_1 = 0{.}5$.}}
\label{ap:mcse}\\
	\toprule
	\multirow{2}{*}{Index} &
	\multirow{2}{*}{$\alpha_2$} &
	\multirow{2}{*}{$\theta$} &
	\multirow{2}{*}{$n$} &
	\multicolumn{5}{c}{Estimator} &
	\multicolumn{5}{c}{Bias-corrected estimator} \\
	\cmidrule(rr){5-9}
	\cmidrule(rr){10-14}
	&&&&
	$\hat{\theta}$ &
	$\mathrm{B}(\hat{\theta})$ &
	$\mathrm{MSE}(\hat{\theta})$ &
	$\mathrm{EP}(\mathrm{B})$ &
	$\mathrm{EP}(\mathrm{MSE})$ &
	$\hat{\theta}$ &
	$\mathrm{B}(\hat{\theta})$ &
	$\mathrm{MSE}(\hat{\theta})$ &
	$\mathrm{EP}(\mathrm{B})$ &
	$\mathrm{EP}(\mathrm{MSE})$ \\
	\midrule
	\multirow{24}{*}{Theil-T} & \multirow{4}{*}{0.5} & \multirow{4}{*}{0.7296} & 10 & 0.6385 & -0.0912 & 0.0652 & 0.0075 & 0.0027 & 0.7083 & -0.0213 & 0.0932 & 0.0096 & 0.0059 \\
	& & & 25 & 0.6941 & -0.0355 & 0.0287 & 0.0052 & 0.0014 & 0.7266 & -0.0031 & 0.0358 & 0.0060 & 0.0023 \\
	& & & 50 & 0.7178 & -0.0118 & 0.0153 & 0.0039 & 0.0007 & 0.7353 & 0.0057 & 0.0176 & 0.0042 & 0.0009 \\
	& & & 100 & 0.7165 & -0.0131 & 0.0073 & 0.0027 & 0.0003 & 0.7253 & -0.0043 & 0.0077 & 0.0028 & 0.0004 \\

	\cmidrule(lr){2-14} & \multirow{4}{*}{1.0} & \multirow{4}{*}{0.6139} & 10 & 0.5407 & -0.0732 & 0.0456 & 0.0063 & 0.0019 & 0.5951 & -0.0188 & 0.0615 & 0.0078 & 0.0038 \\
	& & & 25 & 0.5814 & -0.0325 & 0.0205 & 0.0044 & 0.0009 & 0.6046 & -0.0093 & 0.0237 & 0.0049 & 0.0012 \\
	& & & 50 & 0.6015 & -0.0125 & 0.0105 & 0.0032 & 0.0006 & 0.6141 & 0.0002 & 0.0116 & 0.0034 & 0.0007 \\
	& & & 100 & 0.6075 & -0.0064 & 0.0057 & 0.0024 & 0.0003 & 0.6141 & 0.0002 & 0.0060 & 0.0024 & 0.0003 \\

	\cmidrule(lr){2-14} & \multirow{4}{*}{2.0} & \multirow{4}{*}{0.5858} & 10 & 0.5378 & -0.0479 & 0.0433 & 0.0064 & 0.0021 & 0.5800 & -0.0058 & 0.0583 & 0.0076 & 0.0043 \\
	& & & 25 & 0.5656 & -0.0202 & 0.0177 & 0.0042 & 0.0008 & 0.5824 & -0.0034 & 0.0196 & 0.0044 & 0.0009 \\
	& & & 50 & 0.5797 & -0.0060 & 0.0095 & 0.0031 & 0.0005 & 0.5888 & 0.0030 & 0.0101 & 0.0032 & 0.0005 \\
	& & & 100 & 0.5842 & -0.0016 & 0.0045 & 0.0021 & 0.0002 & 0.5888 & 0.0031 & 0.0047 & 0.0022 & 0.0002 \\

	\cmidrule(lr){2-14} & \multirow{4}{*}{3.0} & \multirow{4}{*}{0.6067} & 10 & 0.5617 & -0.0450 & 0.0441 & 0.0065 & 0.0021 & 0.5996 & -0.0072 & 0.0579 & 0.0076 & 0.0042 \\
	& & & 25 & 0.6015 & -0.0052 & 0.0190 & 0.0044 & 0.0010 & 0.6163 & 0.0095 & 0.0214 & 0.0046 & 0.0013 \\
	& & & 50 & 0.6002 & -0.0066 & 0.0092 & 0.0030 & 0.0004 & 0.6068 & 0.000* & 0.0096 & 0.0031 & 0.0005 \\
	& & & 100 & 0.6030 & -0.0037 & 0.0045 & 0.0021 & 0.0002 & 0.6064 & -0.0003 & 0.0046 & 0.0021 & 0.0002 \\

	\cmidrule(lr){2-14} & \multirow{4}{*}{4.0} & \multirow{4}{*}{0.6322} & 10 & 0.5983 & -0.0339 & 0.0471 & 0.0068 & 0.0023 & 0.6353 & 0.0031 & 0.0635 & 0.0080 & 0.0043 \\
	& & & 25 & 0.6200 & -0.0122 & 0.0179 & 0.0042 & 0.0008 & 0.6305 & -0.0017 & 0.0192 & 0.0044 & 0.0009 \\
	& & & 50 & 0.6295 & -0.0027 & 0.0096 & 0.0031 & 0.0004 & 0.6343 & 0.0020 & 0.0099 & 0.0031 & 0.0005 \\
	& & & 100 & 0.6298 & -0.0024 & 0.0045 & 0.0021 & 0.0002 & 0.6321 & -0.0002 & 0.0046 & 0.0021 & 0.0002 \\

	\cmidrule(lr){2-14} & \multirow{4}{*}{5.0} & \multirow{4}{*}{0.6554} & 10 & 0.6377 & -0.0177 & 0.0483 & 0.0069 & 0.0022 & 0.6712 & 0.0158 & 0.0648 & 0.0080 & 0.0042 \\
	& & & 25 & 0.6493 & -0.0061 & 0.0198 & 0.0044 & 0.0009 & 0.6571 & 0.0016 & 0.0212 & 0.0046 & 0.0010 \\
	& & & 50 & 0.6534 & -0.0021 & 0.0101 & 0.0032 & 0.0005 & 0.6563 & 0.0009 & 0.0103 & 0.0032 & 0.0005 \\
	& & & 100 & 0.6537 & -0.0017 & 0.0049 & 0.0022 & 0.0002 & 0.6551 & -0.0003 & 0.0050 & 0.0022 & 0.0002 \\

	\cmidrule(lr){1-14}
	\multirow{24}{*}{Theil-L} & \multirow{4}{*}{0.5} & \multirow{4}{*}{1.2704} & 10 & 1.1753 & -0.0950 & 0.2930 & 0.0169 & 0.0145 & 1.2681 & -0.0022 & 0.3257 & 0.0181 & 0.0175 \\
	& & & 25 & 1.2346 & -0.0357 & 0.1186 & 0.0108 & 0.0054 & 1.2728 & 0.0025 & 0.1239 & 0.0111 & 0.0060 \\
	& & & 50 & 1.2587 & -0.0116 & 0.0602 & 0.0078 & 0.0026 & 1.2783 & 0.0079 & 0.0618 & 0.0079 & 0.0028 \\
	& & & 100 & 1.2657 & -0.0047 & 0.0296 & 0.0054 & 0.0013 & 1.2753 & 0.0049 & 0.0300 & 0.0055 & 0.0013 \\

	\cmidrule(lr){2-14} & \multirow{4}{*}{1.0} & \multirow{4}{*}{1.0523} & 10 & 0.9570 & -0.0953 & 0.2200 & 0.0145 & 0.0123 & 1.0290 & -0.0233 & 0.2384 & 0.0154 & 0.0144 \\
	& & & 25 & 1.0243 & -0.0280 & 0.0924 & 0.0096 & 0.0043 & 1.0536 & 0.0013 & 0.0955 & 0.0098 & 0.0046 \\
	& & & 50 & 1.0370 & -0.0153 & 0.0438 & 0.0066 & 0.0021 & 1.0520 & -0.0003 & 0.0445 & 0.0067 & 0.0022 \\
	& & & 100 & 1.0470 & -0.0053 & 0.0231 & 0.0048 & 0.0011 & 1.0545 & 0.0022 & 0.0233 & 0.0048 & 0.0011 \\

	\cmidrule(lr){2-14} & \multirow{4}{*}{2.0} & \multirow{4}{*}{1.1043} & 10 & 1.0531 & -0.0512 & 0.2639 & 0.0162 & 0.0136 & 1.1219 & 0.0176 & 0.2903 & 0.0170 & 0.0158 \\
	& & & 25 & 1.0696 & -0.0347 & 0.1021 & 0.0100 & 0.0046 & 1.0962 & -0.0081 & 0.1046 & 0.0102 & 0.0049 \\
	& & & 50 & 1.0927 & -0.0116 & 0.0511 & 0.0071 & 0.0023 & 1.1062 & 0.0019 & 0.0519 & 0.0072 & 0.0024 \\
	& & & 100 & 1.1067 & 0.0024 & 0.0259 & 0.0051 & 0.0012 & 1.1135 & 0.0092 & 0.0261 & 0.0051 & 0.0012 \\

	\cmidrule(lr){2-14} & \multirow{4}{*}{3.0} & \multirow{4}{*}{1.2145} & 10 & 1.1229 & -0.0915 & 0.2618 & 0.0159 & 0.0145 & 1.1937 & -0.0207 & 0.2817 & 0.0168 & 0.0165 \\
	& & & 25 & 1.1883 & -0.0261 & 0.1109 & 0.0105 & 0.0053 & 1.2162 & 0.0018 & 0.1144 & 0.0107 & 0.0057 \\
	& & & 50 & 1.2104 & -0.0040 & 0.0551 & 0.0074 & 0.0031 & 1.2239 & 0.0094 & 0.0561 & 0.0075 & 0.0032 \\
	& & & 100 & 1.2067 & -0.0078 & 0.0268 & 0.0052 & 0.0012 & 1.2134 & -0.0011 & 0.0270 & 0.0052 & 0.0012 \\

	\cmidrule(lr){2-14} & \multirow{4}{*}{4.0} & \multirow{4}{*}{1.3175} & 10 & 1.2327 & -0.0848 & 0.2998 & 0.0171 & 0.0128 & 1.3097 & -0.0078 & 0.3262 & 0.0181 & 0.0153 \\
	& & & 25 & 1.2838 & -0.0337 & 0.1133 & 0.0106 & 0.0047 & 1.3120 & -0.0055 & 0.1161 & 0.0108 & 0.0049 \\
	& & & 50 & 1.3062 & -0.0113 & 0.0587 & 0.0077 & 0.0026 & 1.3200 & 0.0024 & 0.0595 & 0.0077 & 0.0026 \\
	& & & 100 & 1.3134 & -0.0041 & 0.0314 & 0.0056 & 0.0014 & 1.3201 & 0.0026 & 0.0316 & 0.0056 & 0.0014 \\

	\cmidrule(lr){2-14} & \multirow{4}{*}{5.0} & \multirow{4}{*}{1.4086} & 10 & 1.3288 & -0.0797 & 0.3043 & 0.0173 & 0.0137 & 1.4115 & 0.0029 & 0.3335 & 0.0183 & 0.0166 \\
	& & & 25 & 1.3806 & -0.0280 & 0.1177 & 0.0108 & 0.0053 & 1.4096 & 0.0011 & 0.1208 & 0.0110 & 0.0056 \\
	& & & 50 & 1.3992 & -0.0093 & 0.0622 & 0.0079 & 0.0027 & 1.4132 & 0.0046 & 0.0631 & 0.0079 & 0.0028 \\
	& & & 100 & 1.3955 & -0.0131 & 0.0279 & 0.0053 & 0.0012 & 1.4023 & -0.0062 & 0.0280 & 0.0053 & 0.0012 \\

	\cmidrule(lr){1-14}
	\multirow{24}{*}{$\mathrm{A}(1)$} & \multirow{4}{*}{0.5} & \multirow{4}{*}{0.7193} & 10 & 0.6504 & -0.0689 & 0.0312 & 0.0051 & 0.0016 & 0.7130 & -0.0063 & 0.0309 & 0.0056 & 0.0014 \\
	& & & 25 & 0.6924 & -0.0268 & 0.0106 & 0.0032 & 0.0005 & 0.7190 & -0.0003 & 0.0104 & 0.0032 & 0.0005 \\
	& & & 50 & 0.7075 & -0.0117 & 0.0051 & 0.0022 & 0.0002 & 0.7210 & 0.0018 & 0.0050 & 0.0022 & 0.0002 \\
	& & & 100 & 0.7138 & -0.0055 & 0.0024 & 0.0015 & 0.0001 & 0.7206 & 0.0014 & 0.0024 & 0.0015 & 0.0001 \\

	\cmidrule(lr){2-14} & \multirow{4}{*}{1.0} & \multirow{4}{*}{0.6509} & 10 & 0.5791 & -0.0718 & 0.0324 & 0.0052 & 0.0014 & 0.6376 & -0.0132 & 0.0327 & 0.0057 & 0.0013 \\
	& & & 25 & 0.6249 & -0.0260 & 0.0124 & 0.0034 & 0.0006 & 0.6504 & -0.0005 & 0.0126 & 0.0036 & 0.0005 \\
	& & & 50 & 0.6379 & -0.0130 & 0.0056 & 0.0023 & 0.0003 & 0.6508 & -0.0001 & 0.0056 & 0.0024 & 0.0002 \\
	& & & 100 & 0.6450 & -0.0059 & 0.0028 & 0.0017 & 0.0001 & 0.6516 & 0.0007 & 0.0028 & 0.0017 & 0.0001 \\

	\cmidrule(lr){2-14} & \multirow{4}{*}{2.0} & \multirow{4}{*}{0.6686} & 10 & 0.6090 & -0.0596 & 0.0332 & 0.0054 & 0.0015 & 0.6676 & -0.0009 & 0.0351 & 0.0059 & 0.0014 \\
	& & & 25 & 0.6399 & -0.0286 & 0.0126 & 0.0034 & 0.0006 & 0.6645 & -0.0041 & 0.0126 & 0.0035 & 0.0006 \\
	& & & 50 & 0.6564 & -0.0122 & 0.0056 & 0.0023 & 0.0002 & 0.6692 & 0.0007 & 0.0056 & 0.0024 & 0.0002 \\
	& & & 100 & 0.6651 & -0.0035 & 0.0029 & 0.0017 & 0.0001 & 0.6716 & 0.0031 & 0.0029 & 0.0017 & 0.0001 \\

	\cmidrule(lr){2-14} & \multirow{4}{*}{3.0} & \multirow{4}{*}{0.7031} & 10 & 0.6363 & -0.0668 & 0.0303 & 0.0051 & 0.0015 & 0.6946 & -0.0086 & 0.0300 & 0.0055 & 0.0014 \\
	& & & 25 & 0.6791 & -0.0241 & 0.0105 & 0.0032 & 0.0005 & 0.7033 & 0.0002 & 0.0104 & 0.0032 & 0.0005 \\
	& & & 50 & 0.6939 & -0.0092 & 0.0048 & 0.0022 & 0.0002 & 0.7062 & 0.0031 & 0.0049 & 0.0022 & 0.0002 \\
	& & & 100 & 0.6968 & -0.0063 & 0.0025 & 0.0016 & 0.0001 & 0.7030 & -0.0001 & 0.0025 & 0.0016 & 0.0001 \\

	\cmidrule(lr){2-14} & \multirow{4}{*}{4.0} & \multirow{4}{*}{0.7322} & 10 & 0.6685 & -0.0637 & 0.0282 & 0.0049 & 0.0013 & 0.7263 & -0.0059 & 0.0274 & 0.0052 & 0.0011 \\
	& & & 25 & 0.7076 & -0.0246 & 0.0097 & 0.0030 & 0.0005 & 0.7306 & -0.0016 & 0.0094 & 0.0031 & 0.0004 \\
	& & & 50 & 0.7212 & -0.0110 & 0.0046 & 0.0021 & 0.0002 & 0.7327 & 0.0005 & 0.0046 & 0.0021 & 0.0002 \\
	& & & 100 & 0.7269 & -0.0053 & 0.0024 & 0.0015 & 0.0001 & 0.7326 & 0.0004 & 0.0024 & 0.0015 & 0.0001 \\

	\cmidrule(lr){2-14} & \multirow{4}{*}{5.0} & \multirow{4}{*}{0.7555} & 10 & 0.6972 & -0.0583 & 0.0259 & 0.0047 & 0.0015 & 0.7529 & -0.0026 & 0.0249 & 0.0050 & 0.0013 \\
	& & & 25 & 0.7341 & -0.0214 & 0.0081 & 0.0028 & 0.0005 & 0.7564 & 0.0008 & 0.0079 & 0.0028 & 0.0004 \\
	& & & 50 & 0.7456 & -0.0099 & 0.0039 & 0.0020 & 0.0002 & 0.7566 & 0.0011 & 0.0038 & 0.0020 & 0.0002 \\
	& & & 100 & 0.7489 & -0.0067 & 0.0018 & 0.0013 & 0.0001 & 0.7544 & -0.0012 & 0.0017 & 0.0013 & 0.0001 \\

	\cmidrule(lr){1-14}
	\multirow{24}{*}{$\mathrm{A}(\infty)$} & \multirow{4}{*}{0.5} & \multirow{4}{*}{1.0000} & 10 & 0.9741 & -0.0259 & 0.0028 & 0.0015 & 0.0004 & 1.0290 & 0.0290 & 0.0015 & 0.0008 & 0.0001 \\
	& & & 25 & 0.9953 & -0.0047 & 0.0001 & 0.0003 & 0.000* & 1.0047 & 0.0047 & 0.0001 & 0.0002 & 0.000* \\
	& & & 50 & 0.9988 & -0.0012 & 0.000* & 0.0001 & 0.000* & 1.0010 & 0.0010 & 0.000* & 0.0001 & 0.000* \\
	& & & 100 & 0.9997 & -0.0003 & 0.000* & 0.000* & 0.000* & 1.0002 & 0.0002 & 0.000* & 0.000* & 0.000* \\

	\cmidrule(lr){2-14} & \multirow{4}{*}{1.0} & \multirow{4}{*}{1.0000} & 10 & 0.9635 & -0.0365 & 0.0040 & 0.0016 & 0.0003 & 1.0383 & 0.0383 & 0.0024 & 0.0010 & 0.0001 \\
	& & & 25 & 0.9922 & -0.0078 & 0.0003 & 0.0005 & 0.000* & 1.0075 & 0.0075 & 0.0001 & 0.0003 & 0.000* \\
	& & & 50 & 0.9978 & -0.0022 & 0.000* & 0.0001 & 0.000* & 1.0019 & 0.0019 & 0.000* & 0.0001 & 0.000* \\
	& & & 100 & 0.9995 & -0.0005 & 0.000* & 0.000* & 0.000* & 1.0005 & 0.0005 & 0.000* & 0.000* & 0.000* \\

	\cmidrule(lr){2-14} & \multirow{4}{*}{2.0} & \multirow{4}{*}{1.0000} & 10 & 0.9696 & -0.0304 & 0.0041 & 0.0018 & 0.0005 & 1.0333 & 0.0333 & 0.0019 & 0.0009 & 0.0001 \\
	& & & 25 & 0.9943 & -0.0057 & 0.0002 & 0.0003 & 0.000* & 1.0063 & 0.0063 & 0.0001 & 0.0003 & 0.000* \\
	& & & 50 & 0.9985 & -0.0015 & 0.000* & 0.0001 & 0.000* & 1.0013 & 0.0013 & 0.000* & 0.0001 & 0.000* \\
	& & & 100 & 0.9996 & -0.0004 & 0.000* & 0.000* & 0.000* & 1.0002 & 0.0002 & 0.000* & 0.000* & 0.000* \\

	\cmidrule(lr){2-14} & \multirow{4}{*}{3.0} & \multirow{4}{*}{1.0000} & 10 & 0.9774 & -0.0226 & 0.0022 & 0.0013 & 0.0003 & 1.0295 & 0.0295 & 0.0016 & 0.0009 & 0.0001 \\
	& & & 25 & 0.9959 & -0.0041 & 0.0001 & 0.0002 & 0.000* & 1.0042 & 0.0042 & 0.000* & 0.0002 & 0.000* \\
	& & & 50 & 0.9990 & -0.0010 & 0.000* & 0.0001 & 0.000* & 1.0008 & 0.0008 & 0.000* & 0.0001 & 0.000* \\
	& & & 100 & 0.9997 & -0.0003 & 0.000* & 0.000* & 0.000* & 1.0002 & 0.0002 & 0.000* & 0.000* & 0.000* \\

	\cmidrule(lr){2-14} & \multirow{4}{*}{4.0} & \multirow{4}{*}{1.0000} & 10 & 0.9823 & -0.0177 & 0.0012 & 0.0009 & 0.0002 & 1.0251 & 0.0251 & 0.0012 & 0.0007 & 0.0001 \\
	& & & 25 & 0.9965 & -0.0035 & 0.0001 & 0.0002 & 0.000* & 1.0031 & 0.0031 & 0.000* & 0.0002 & 0.000* \\
	& & & 50 & 0.9992 & -0.0008 & 0.000* & 0.0001 & 0.000* & 1.0007 & 0.0007 & 0.000* & 0.000* & 0.000* \\
	& & & 100 & 0.9998 & -0.0002 & 0.000* & 0.000* & 0.000* & 1.0001 & 0.0001 & 0.000* & 0.000* & 0.000* \\

	\cmidrule(lr){2-14} & \multirow{4}{*}{5.0} & \multirow{4}{*}{1.0000} & 10 & 0.9840 & -0.0160 & 0.0012 & 0.0010 & 0.0002 & 1.0201 & 0.0201 & 0.0008 & 0.0006 & 0.0001 \\
	& & & 25 & 0.9971 & -0.0029 & 0.000* & 0.0002 & 0.000* & 1.0023 & 0.0023 & 0.000* & 0.0001 & 0.000* \\
	& & & 50 & 0.9993 & -0.0007 & 0.000* & 0.000* & 0.000* & 1.0005 & 0.0005 & 0.000* & 0.000* & 0.000* \\
	& & & 100 & 0.9998 & -0.0002 & 0.000* & 0.000* & 0.000* & 1.0001 & 0.0001 & 0.000* & 0.000* & 0.000* \\

	\cmidrule(lr){1-14}
	\multirow{24}{*}{$\mathrm{VMR}$} & \multirow{4}{*}{0.5} & \multirow{4}{*}{1.0000} & 10 & 0.8609 & -0.1391 & 0.4264 & 0.0202 & 0.0369 & 0.9544 & -0.0456 & 0.6198 & 0.0249 & 0.0687 \\
	& & & 25 & 0.9286 & -0.0714 & 0.2351 & 0.0152 & 0.0277 & 0.9866 & -0.0134 & 0.3190 & 0.0179 & 0.0488 \\
	& & & 50 & 0.9718 & -0.0282 & 0.1342 & 0.0116 & 0.0117 & 1.0053 & 0.0053 & 0.1611 & 0.0127 & 0.0161 \\
	& & & 100 & 0.9728 & -0.0272 & 0.0696 & 0.0083 & 0.0047 & 0.9899 & -0.0101 & 0.0763 & 0.0087 & 0.0058 \\

	\cmidrule(lr){2-14} & \multirow{4}{*}{1.0} & \multirow{4}{*}{1.0857} & 10 & 0.9243 & -0.1614 & 0.3815 & 0.0189 & 0.0251 & 1.0056 & -0.0801 & 0.5317 & 0.0229 & 0.0459 \\
	& & & 25 & 1.0237 & -0.0620 & 0.2311 & 0.0151 & 0.0211 & 1.0691 & -0.0167 & 0.2874 & 0.0170 & 0.0342 \\
	& & & 50 & 1.0696 & -0.0162 & 0.1308 & 0.0114 & 0.0085 & 1.0957 & 0.0100 & 0.1479 & 0.0122 & 0.0108 \\
	& & & 100 & 1.0690 & -0.0167 & 0.0688 & 0.0083 & 0.0035 & 1.0828 & -0.0029 & 0.0733 & 0.0086 & 0.0040 \\

	\cmidrule(lr){2-14} & \multirow{4}{*}{2.0} & \multirow{4}{*}{1.4909} & 10 & 1.3737 & -0.1172 & 0.6434 & 0.0251 & 0.0400 & 1.4621 & -0.0288 & 0.8866 & 0.0298 & 0.0726 \\
	& & & 25 & 1.4096 & -0.0813 & 0.2791 & 0.0165 & 0.0151 & 1.4505 & -0.0404 & 0.3193 & 0.0178 & 0.0193 \\
	& & & 50 & 1.4773 & -0.0136 & 0.1705 & 0.0131 & 0.0118 & 1.5016 & 0.0107 & 0.1883 & 0.0137 & 0.0146 \\
	& & & 100 & 1.4949 & 0.0040 & 0.0821 & 0.0091 & 0.0038 & 1.5075 & 0.0165 & 0.0858 & 0.0093 & 0.0041 \\

	\cmidrule(lr){2-14} & \multirow{4}{*}{3.0} & \multirow{4}{*}{2.0000} & 10 & 1.8531 & -0.1469 & 0.8766 & 0.0293 & 0.0626 & 1.9453 & -0.0547 & 1.1630 & 0.0341 & 0.1111 \\
	& & & 25 & 1.9939 & -0.0061 & 0.4104 & 0.0203 & 0.0270 & 2.0413 & 0.0413 & 0.4813 & 0.0219 & 0.0353 \\
	& & & 50 & 1.9738 & -0.0262 & 0.2014 & 0.0142 & 0.0106 & 1.9946 & -0.0054 & 0.2147 & 0.0147 & 0.0118 \\
	& & & 100 & 1.9922 & -0.0078 & 0.0999 & 0.0100 & 0.0049 & 2.0029 & 0.0029 & 0.1031 & 0.0102 & 0.0052 \\

	\cmidrule(lr){2-14} & \multirow{4}{*}{4.0} & \multirow{4}{*}{2.5474} & 10 & 2.3781 & -0.1692 & 1.0873 & 0.0326 & 0.0677 & 2.4771 & -0.0703 & 1.4307 & 0.0378 & 0.1189 \\
	& & & 25 & 2.5114 & -0.0360 & 0.4675 & 0.0216 & 0.0268 & 2.5467 & -0.0007 & 0.5302 & 0.0230 & 0.0365 \\
	& & & 50 & 2.5421 & -0.0052 & 0.2464 & 0.0157 & 0.0118 & 2.5574 & 0.0100 & 0.2576 & 0.0161 & 0.0128 \\
	& & & 100 & 2.5503 & 0.0030 & 0.1189 & 0.0109 & 0.0058 & 2.5572 & 0.0099 & 0.1213 & 0.0110 & 0.0059 \\
	\cmidrule(lr){2-14} & \multirow{4}{*}{5.0} & \multirow{4}{*}{3.1130} & 10 & 3.0522 & -0.0608 & 1.5423 & 0.0392 & 0.0942 & 3.1433 & 0.0302 & 1.9668 & 0.0444 & 0.1688 \\
	& & & 25 & 3.0927 & -0.0203 & 0.5514 & 0.0235 & 0.0313 & 3.1158 & 0.0027 & 0.6020 & 0.0245 & 0.0353 \\
	& & & 50 & 3.1020 & -0.0110 & 0.3029 & 0.0174 & 0.0154 & 3.1091 & -0.0039 & 0.3124 & 0.0177 & 0.0161 \\
	& & & 100 & 3.1188 & 0.0058 & 0.1419 & 0.0119 & 0.0074 & 3.1218 & 0.0088 & 0.1441 & 0.0120 & 0.0076 \\
	\bottomrule
\end{longtable}
\vspace*{-8pt}
\hspace*{-20pt}
{\small * True value below the precision limit ($<0{.}0001$).}

}

\clearpage

\section{Application to real data}\label{cap:dados-reais}

The data used correspond to the Gross Domestic Product (GDP) per capita\footnote{Indicator *GDP (Gross Domestic Product) per capita, PPP (constant 2021 international \$)*, which measures per-capita GDP adjusted by Purchasing Power Parity (PPP) and expressed in constant 2021 international dollars. Source: *World Bank Open Data*~\citep{WorldbankGDP2024}.} for countries in North America and Oceania. We considered observations for the most recent available year (2023), totaling $n = 18$ countries. Values were expressed in thousands of dollars to facilitate interpretation and comparison across the economies analyzed.

Figure~\ref{fig:hist-pib} shows a histogram for the GDP data set. It reveals a right-skewed distribution, with most countries concentrated in the lower GDP-per-capita range. The gap between the mean (22.99 thousand dollars) and the median (9.68 thousand dollars) confirms this asymmetry, while the high standard deviation (25.50 thousand dollars) indicates substantial dispersion and reinforces the heterogeneity among economies. A mild bimodality is also observed, with one large group of lower-income countries and a smaller group composed of developed economies.

\begin{figure}[!ht]
\centering
\includegraphics[width=0.60\textwidth]{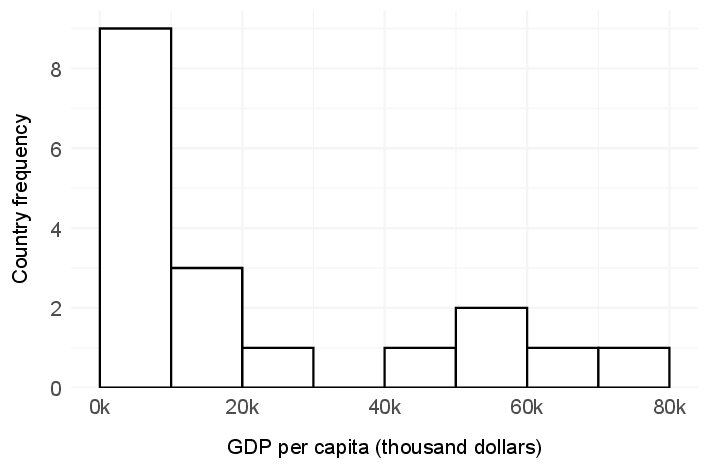}
\caption{Distribution of GDP per capita (in thousands of dollars) for countries in North America and Oceania in 2023.}
\label{fig:hist-pib}
\end{figure}

Figure~\ref{fig:hist-pib-with-curve} highlights this bimodal pattern through the fit of a mixture of gamma distributions, which is suitable for capturing both the separation between groups and the internal asymmetry of each component. The estimated parameters are: mixture proportions $\hat{\pi}_1 = 0.7221$ and $\hat{\pi}_2 = 0.2779$, shape parameters $\hat{\alpha}_1 = 2.7966$ and $\hat{\alpha}_2 = 21.9126$, and a common rate parameter $\hat{\lambda} = 0.3528$. The estimated means and standard deviations of each component, computed as $\hat{\mu}_j = \hat{\alpha}_j / \hat{\lambda}$ and $\hat{\sigma}_j = \sqrt{\hat{\alpha}_j}/\hat{\lambda}$, are approximately 9.9 thousand and 5.9 thousand dollars for the first group and 77.4 thousand and 16.5 thousand dollars for the second. These results indicate that about 72\% of countries belong to a lower-income group with higher relative variability, while the remaining countries form a more homogeneous, high-income cluster, reflecting the structural duality in the distribution.

\begin{figure}[!ht]
\centering
\includegraphics[width=0.60\textwidth]{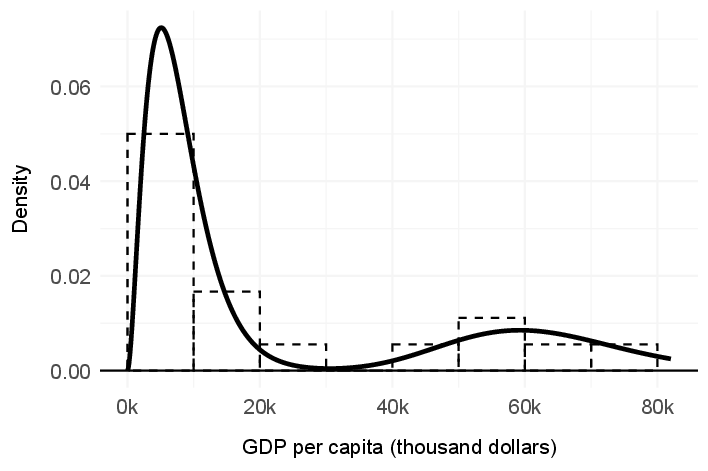}
\caption{Empirical distribution of GDP per capita (in thousands of dollars) for countries in North America and Oceania in 2023, showing the bimodality captured by the fitted curve.}
\label{fig:hist-pib-with-curve}
\end{figure}

Based on the fitted mixture model, empirical estimates of selected inequality indices were obtained. For each index, we computed the standard estimator $\hat{\theta}$, the estimated analytical bias $\widehat{\mathrm{Bias}}$, evaluated at the fitted model parameters, and the bias-corrected estimator $\hat{\theta}_{\mathrm{corr}} = \hat{\theta} - \widehat{\mathrm{Bias}}$. The results are reported in Table~\ref{tab:empirical-bias}, whose joint interpretation allows for a comprehensive assessment of the observed inequality.

\begin{table}[h]
\centering
\caption{Empirical estimates of inequality indices and analytical bias corrections under a fitted gamma mixture model based on the GDP data set.}
\label{tab:empirical-bias}
\begin{tabular}{lrrr}
\toprule
\textbf{Index} & $\widehat{\theta}$ & $\widehat{\text{Bias}}$ & $\widehat{\theta}_{corr}$ \\
\midrule
Theil-T & $0.5282$ & $-0.0246$ & $0.5529$ \\
Theil-L & $0.6354$ & $-0.0364$ & $0.6718$ \\
Atkinson ($\varepsilon=1$) & $0.4703$ & $-0.0233$ & $0.4936$ \\
Atkinson ($\varepsilon\to\infty$) & $0.8898$ & $-0.1002$ & $0.9899$ \\
VMR & $28.2871$ & $-0.3068$ & $28.5939$ \\
\bottomrule
\end{tabular}
\end{table}

From Table~\ref{tab:empirical-bias}, we observe that among the measures analyzed, the indices Theil-T and Theil-L again stand out, with corrected values of $0.5529$ and $0.6718$, respectively. Their relative proximity indicates that most inequality is explained by the lower ranges of the income distribution, dominated by low- and middle-income countries. The index Theil-L, which assigns greater weight to lower incomes, shows a slightly higher value than Theil-T, which measures divergence around the mean and captures differences among middle and upper strata as well. The Atkinson index with $\varepsilon = 1$ has a corrected value of $0.4936$, implying that a perfectly equal distribution preserving the same level of social welfare would correspond to about 50.6\% of the global average income. In the limit case $\varepsilon \to \infty$, the corrected value of $0.9899$ suggests an equally distributed equivalent income near 1\% of the mean. Although extreme in magnitude, this result follows directly from the infinite sensitivity to the lowest incomes and is consistent with the strong asymmetry of the sample: the median (9.68 thousand dollars) is about one-third of the mean (22.99 thousand dollars), and the poorer countries cluster near the lower bound of the fitted distribution. Thus, the high value of $A(\infty)$ does not indicate model inconsistency; rather, it reflects the severe penalization of the lower tail under extreme inequality aversion.

The bias estimates shown in Table~\ref{tab:empirical-bias} were obtained from analytical corrections evaluated at the fitted mixture parameters—that is, they are plug-in corrections: the estimated parameters are substituted for the true ones in the theoretical bias expression. Although this procedure is consistent and widely used, it does not propagate the additional uncertainty arising from model fitting. Therefore, the corrected value should be interpreted as an approximation to the unbiased estimator under the fitted model, rather than as an exact correction in the strict frequentist sense.

The estimated dispersion index reinforces the conclusions drawn from the entropy-based measures. However, interpreting $\widehat{\mathrm{VMR}}_{\text{corr}} = 28.5939$ requires care: unlike entropy-based indices, the VMR is not scale-invariant. Since $\mathrm{Var}(X)$ and $\mathbb{E}[X]$ carry different units, the ratio $\mathrm{Var}(X)/\mathbb{E}[X]$ inherits the unit of $X$. This means that simply changing the unit from “thousand dollars” to “dollars” would multiply all VMR estimates by 1,000, without altering any substantive feature of the distribution but changing its numerical magnitude. Thus, the observed value reflects the chosen monetary unit and should not be directly compared with results expressed on different scales.

Finally, the analytical bias corrections were on the order of $10^{-2}$, showing that the empirical estimators are close to their unbiased counterparts under the fitted model; nevertheless, applying the corrections proved essential for removing residual bias and empirically validating the theoretical expressions derived in Section~\ref{Deriving estimator biases}.

\section{Concluding remarks}\label{concluding_remarks}

This study investigated the behavior of the Theil, Atkinson, and VMR inequality estimators under population heterogeneity modeled by finite mixtures of gamma distributions with a common rate parameter. Closed-form expressions for the expected value of plug-in estimators were obtained using Mosimann’s proportion–sum independence theorem and core properties of the Dirichlet distribution, allowing us to formally characterize how bias arises from the interaction between mixture components. Monte Carlo simulation results have shown that, in general, the absolute bias decreases as sample size grows, confirming the consistency of the estimators. However, the effect of heterogeneity is not uniform across inequality measures. Even so, the indices examined here exhibited predominantly negative bias, whereas the Gini coefficient estimator, according to \citet{VilaSaulo2025GiniMixture}, tends to display positive bias under mixture structures. The proposed corrections reduced the bias, especially in small and moderate samples, while differences in mean squared error were modest and diminished as $n$ increased. The empirical application to global per capita GDP data confirmed both the adequacy of gamma mixtures for capturing structural heterogeneity and the practical relevance of bias corrections. Important limitations include the assumption of a common rate parameter, the plug-in nature of the corrections, and the scale dependence of the VMR index. As part of future work, it would be of interest to extend the methodology to mixtures with distinct rate parameters, to develop bootstrap-based bias corrections, and to explore multivariate generalizations.  Work on these problems is currently under progress and we hope to report the findings in a future paper.

\section*{Acknowledgments}

This study was financed in part by the Coordenação de Aperfeiçoamento de Pessoal de Nível Superior (CAPES), Brazil – Finance Code 001. Helton Saulo gratefully acknowledges financial support from: (ii) the University of Brasília; and (ii) the Brazilian National Council for Scientific and Technological Development (CNPq), through Grant 304716/2023-5.

\section*{Declaration}

There are no conflicts of interest to disclose.




\begin{thebibliography}{}
	
	\bibitem[Atkinson, 1970]{Atkinson1970}
	Atkinson, A.~B. (1970).
	\newblock On the measurement of inequality.
	\newblock {\em Journal of Economic Theory}, 2(3):244--263.
	
	\bibitem[Baydil et~al., 2025]{BaydilEtAl2025}
	Baydil, B., de~la Peña, V.~H., Zou, H., and Yao, H. (2025).
	\newblock {Unbiased estimation of the Gini coefficient}.
	\newblock {\em Statistics \& Probability Letters}, 222(C).
	
	\bibitem[Bourguignon, 1979]{bourguignon1979}
	Bourguignon, F. (1979).
	\newblock Decomposable income inequality measures.
	\newblock {\em Econometrica: Journal of the Econometric Society},
	47(4):901--920.
	
	\bibitem[Brezis, 2010]{Brezis2010}
	Brezis, H. (2010).
	\newblock {\em Functional Analysis, Sobolev Spaces and Partial Differential
		Equations}.
	\newblock Editora Springer Science \& Business Media, New York.
	
	\bibitem[Bullen, 2003]{Bullen2003}
	Bullen, P.~S. (2003).
	\newblock {\em Chapter III - The Power Means. Handbook of Means and Their
		Inequalities}.
	\newblock Dordrecht, Netherlands: Kluwer, pp. 175–265.
	
	\bibitem[Chakravarty, 2009]{Chakravarty2008}
	Chakravarty, S.~R. (2009).
	\newblock {\em Inequality, polarization and poverty: Advances in distributional
		analysis}.
	\newblock Springer.
	
	\bibitem[Chotikapanich and Griffiths, 2008]{chotikapanich2008}
	Chotikapanich, D. and Griffiths, W.~E. (2008).
	\newblock {\em Estimating Income Distributions Using a Mixture of Gamma
		Densities}, pages 285--302.
	\newblock Springer New York, New York, NY.
	
	\bibitem[Cowell, 2000]{Cowell2000}
	Cowell, F.~A. (2000).
	\newblock Measurement of inequality.
	\newblock {\em Handbook of income distribution}, 1:87--166.
	
	\bibitem[Deltas, 2003]{Deltas2003}
	Deltas, G. (2003).
	\newblock The small-sample bias of the gini coefficient: results and
	implications for empirical research.
	\newblock {\em Review of Economics and Statistics}, 85:226--234.
	
	\bibitem[Frühwirth-Schnatter, 2006]{fruhwirth2006finite}
	Frühwirth-Schnatter, S. (2006).
	\newblock {\em Finite Mixture and Markov Switching Models}.
	\newblock Springer, New York, 1st edition.
	
	\bibitem[Gini, 1912]{Gini1912}
	Gini, C. (1912).
	\newblock Variabilit{\`a} e mutabilit{\`a}: contributo allo studio delle
	distribuzioni e delle relazioni statistiche. [fasc. i.].
	
	\bibitem[Kai Wang~Ng, 2011]{Kai2011}
	Kai Wang~Ng, Guo-Liang~Tian, M.-L.~T. (2011).
	\newblock {\em Dirichlet and Related Distributions: Theory, Methods and
		Applications}.
	\newblock John Wiley \& Sons.
	
	\bibitem[McDonald and Jensen, 1979]{McDonald1979}
	McDonald, J.~B. and Jensen, B.~C. (1979).
	\newblock An analysis of some properties of alternative measures of income
	inequality based on the gamma distribution function.
	\newblock {\em Journal of the American Statistical Association}, 74:856--860.
	
	\bibitem[McLachlan and Peel, 2000]{mclachlan2000finite}
	McLachlan, G.~J. and Peel, D. (2000).
	\newblock {\em Finite Mixture Models}.
	\newblock Wiley-Interscience, New York, 1st edition.
	
	\bibitem[Mosimann, 1962]{Mosimann1962}
	Mosimann, J.~E. (1962).
	\newblock On the compound multinomial distribution, the multivariate
	$\beta$-distribution, and correlations among proportions.
	\newblock {\em Biometrika}, 49:65--82.
	
	\bibitem[{R Core Team}, 2024]{RCore2024}
	{R Core Team} (2024).
	\newblock {\em R: A Language and Environment for Statistical Computing}.
	\newblock R Foundation for Statistical Computing, Vienna, Austria.
	\newblock Version 4.4.2 (2024-10-31).
	
	\bibitem[Salem and Mount, 1974]{salem1974}
	Salem, A. B.~Z. and Mount, T.~D. (1974).
	\newblock A convenient descriptive model of income distribution: The gamma
	density.
	\newblock {\em Econometrica: Journal of the Econometric Society},
	42(6):1115--1127.
	
	\bibitem[Sen, 1973]{sen1973}
	Sen, A. (1973).
	\newblock {\em On Economic Inequality}.
	\newblock Oxford University Press, Oxford.
	
	\bibitem[Shih, 2025]{Shih2025}
	Shih, J.-H. (2025).
	\newblock Expectations of some ratio-type estimators under the gamma
	distribution.
	\newblock Manuscript submitted for publication. Available at
	https://arxiv.org/pdf/2505.05080.
	
	\bibitem[Shorrocks, 1980]{Shorrocks1980}
	Shorrocks, A.~F. (1980).
	\newblock The class of additively decomposable inequality measures.
	\newblock {\em Econometrica: Journal of the Econometric Society},
	48(3):613--625.
	
	\bibitem[Theil, 1967]{Theil1967}
	Theil, H. (1967).
	\newblock {\em Economics and Information Theory}.
	\newblock Rand McNally and Company, Chicago.
	
	\bibitem[Vila and Saulo, 2025a]{VilaSaulo2025GiniMixture}
	Vila, R. and Saulo, H. (2025a).
	\newblock Bias in gini coefficient estimation for gamma mixture populations.
	\newblock {\em Statistical Papers}, 66(7):146.
	
	\bibitem[Vila and Saulo, 2025b]{VilaSaulo2025TheilAtkinsonGamma}
	Vila, R. and Saulo, H. (2025b).
	\newblock Closed-form formulas for the biases of the theil and atkinson index
	estimators in gamma distributed populations.
	\newblock {\em arXiv preprint arXiv:2504.13806}.
	
	\bibitem[Vila and Saulo, 2025c]{VilaSaulo2025GiniM}
	Vila, R. and Saulo, H. (2025c).
	\newblock The mth gini index estimator: Unbiasedness for gamma populations.
	\newblock {\em Journal of Economic Inequality (to appear)}.
	
	\bibitem[{World Bank}, 2024]{WorldbankGDP2024}
	{World Bank} (2024).
	\newblock Gdp per capita, ppp (constant 2021 international dollars).
	\newblock International Comparison Program (ICP), World Bank; Eurostat PPP
	Program; OECD PPP Program; IMF World Economic Outlook database; National
	Accounts files. Data published on May 30, 2024. License: CC BY-4.0.
	
\end{thebibliography}

\end{document}